\documentclass[%
reprint,
nofootinbib,
amsmath,amssymb,
aps,
floatfix,
]{revtex4-1}

\usepackage{graphicx}
\usepackage{dcolumn}
\usepackage{bm}
\usepackage{amsmath ,amssymb ,amsfonts,mathrsfs,wasysym,amsthm, textcomp, mathtools,graphicx,stmaryrd,lipsum}
\DeclareMathAlphabet{\mathpzc}{OT1}{pzc}{m}{it}
\usepackage{hyperref}

\hypersetup{%
	colorlinks=false,
	pdfborderstyle={/S/U/W 2}
}

\newcommand{\A}{\mathcal{A}}
\newcommand{\B}{\mathcal{B}}
\newcommand{\Rea}{\mathbb{R}}
\newcommand{\Comp}{\mathbb{C}}
\newcommand{\Id}{\hat{\mathbb{I}}}
\newcommand{\Nat}{\mathbb{N}}

\newcommand{\Ho}{\mathcal{H}_\omega}
\newcommand{\Hi}{\mathcal{H}}

\newcommand{\boundo}{\mathcal{B}(\mathcal{H})}

\newcommand{\boundoo}{\mathcal{B}(\mathcal{H}_{\omega})}

\newcommand{\tco}{\mathcal{B}_{1}(\mathcal{H})}

\newcommand{\borel}{\mathscr{B}(\mathbb{R})}
\newcommand{\borell}[1]{\mathscr{B}(\mathbb{R}^{#1})}
\newcommand{\borelx}{\mathscr{B}(\mathsf{X})}

\newcommand{\unit}{\mathsf{I}}
\newcommand{\E}{\mathcal{E}}

\newcommand{\Ex}{\mathbb{E}}

\newcommand{\Tr}[1]{\mbox{Tr}\left[#1\right]}

\newcommand{\X}{\mathsf{X}}
\newcommand{\topo}{(\X,\mathscr {T})}

\newcommand{\PS}{(\Omega,\E,P)}

\newcommand{\vona}{\mathcal{V}(\Hi)}

\newcommand{\Int}{\mathbb{Z}}

\begin{document}
	
	\newtheorem{definition}{Definition}
	\newtheorem{proposition}{Proposition}
	\newtheorem{theorem}{Theorem}
	\newtheorem{lemma}{Lemma}
	\newtheorem{corollary}{Corollary}
	\newtheorem{assumptions}{Assumptions}
	\newtheorem{assumption}{Assumption}
	
	\renewcommand{\thefootnote}{\fnsymbol{footnote}}
	
	\preprint{}
	
	\title{On non-commutativity in quantum theory (I): \\ from classical to quantum probability.}
	
	\author{Curcuraci Luca}
	\affiliation{Department of Physics, University of Trieste, Strada Costiera 11 34151, Trieste, Italy \\ Istituto Nazionale di Fisica Nucleare, Trieste Section, Via Valerio 2 34127, Trieste, Italy}
	\email{ Curcuraci.article@protonmail.com; \\
		luca.curcuraci@phd.units.it }
	
	\date{\today}
	
	\date{\today}
	
	\begin{abstract}
		
		
		A central feature of quantum mechanics is the non-commutativity of operators used to describe physical observables. In this article, we present a critical analysis on the role of non-commutativity in quantum theory, focusing on its consequences in the probabilistic description. Typically, a random phenomenon is described using the measure-theoretic formulation of probability theory. Such a description can also be done using algebraic methods, which are capable to deal with non-commutative random variables (like in quantum mechanics). Here we propose a method to construct a non-commutative probability theory starting from an ordinary measure-theoretic description of probability. This will be done using the entropic uncertainty relations between random variables, in order to evaluate the presence of non-commutativity in their algebraic description. 
		
	\end{abstract}
	
	\keywords{Quantum mechanics, Probability, Quantum probability Algebraic probability, Entropic uncertainty relations, Non-commutativity.}
	\maketitle
	
	\tableofcontents
	
	\section{Introduction}
	
	Measure-theoretic and algebraic probability theory na\"{i}vely correspond to classical and quantum probability, respectively. They were born almost in the same years, the first from the pioneering works of A. Kolmogorov, while the second with the works of J. von Neumann. Nevertheless, they remained separated for many years. The measure-theoretic approach seems to be applicable to any \textquotedblleft ordinary/everyday random phenomena", while the algebraic approach was originally motivated to describe \textquotedblleft quantum random phenomena": two aspects of randomness which are considered deeply different.
	
	The measure-theoretic approach to probability seems to privilege the probability while the algebraic approach is more about random variables (i.e. the mathematical representation of features of a random phenomenon). Despite they appear as two completely different approaches, they are intimately related. Typically people consider the first as a special case of the second, but this point of view could be too simple. The recent interest in quantum information arises questions about the meaning of the quantization procedure, which may be seen as a change in the probabilistic description of natural random phenomena. Nevertheless, this change in the description of random phenomena is suggested also in other fields which are not really related to the quantum world, like for example finance\cite{segal1998black,meng2016quantum,meng2015quantum},  social science (for a very interesting experiment, see Ch.6 in \cite{khrennikov2014ubiquitous}), or cognitive science \cite{khrennikov2014ubiquitous,aerts1997applications}.
	
	The purpose of this article is to present a method which allows to evaluate the presence of non-commutativity in the probabilistic description of a given random phenomena. Such method will be used in \cite{LC2} and \cite{LC3} to explicitly construct models which reproduce the non-commutativity between position and velocity of a particle, as in non-relativistic quantum mechanics. In section \ref{SECT1}, the ordinary measure-theoretic approach to probability will be reconsidered in an Hilbert-space setting. A review of the basic facts about algebraic (in general non-commutative) probability spaces will be done in section \ref{section2}, which allows us to present the so called entropic uncertainty relations under a different light in section \ref{section4}. Finally, a method for the construction of a non-commutative probability space starting from a collection of ordinary probability spaces using entropic uncertainty relations will be presented in section \ref{section5}.

	\section{Classical probability in Hilbert space language}\label{SECT1}
	
	In this section, we collect a series of results about measure-theoretic probability theory and its algebraic formulation \cite{massen1998quantum}. We will start reviewing the standard measure-theoretic probability, pointing out the algebraic structures which are already present in the ordinary formulation. After the introduction of the necessary mathematical tools \cite{moretti2013spectral}, we will show how standard measure-theoretic probability look like in Hilbert spaces.
	
	\subsection{Measure-theoretic probability, i.e. classical probability}\label{Sec1a}
	
	With the term \emph{classical probability} we will refer to Kolmogorov's formulation of probability theory based on measure space. According to this framework, the description of a random phenomenon is made by using the triple $\PS$, called \emph{probability space}, where
	\begin{enumerate}
		\item[i)] $\Omega$ is the \emph{sample space}, and it represents the set of all possibile elementary outcomes of a random experiment;
		\item[ii)] $\E$ is a \emph{$\sigma$-algebra on $\Omega$}, namely a collection of subsets of $\Omega$ which is closed under complement and countable union. It can be understood as the set of all propositions (also called \emph{events}) about the random phenomenon whose truth value can be tested with an experiment;
		\item[iii)] $P$ is a \emph{probability measure}, namely a map $P:\E \rightarrow [0,1]$, which is normalised ($P(\Omega) = 1$) and $\sigma$-additive. If $A\in \E$ is an event, then we will interpret $P(A)$ as the degree of belief that the event $A$ really happens. 
	\end{enumerate}
	Note that in Kolmogorov's formulation, a probability space is nothing but a measure space where the measure is normalised. A random variable $X$ in this context, i.e. a feature of a random phenomenon, is simply any $P$-measurable map from the probability space to another measurable space\footnote{A measurable space is the couple $(M,\mathcal{M})$ where $M$ is a set and $\mathcal{M}$ is a $\sigma$-algebra on this set. When equipped with a measure it becomes a measure space.} $(M,\mathcal{M})$, hence $X: \PS \rightarrow (M,\mathcal{M})$. The image of the probability measure $P$, under the map $X$, induces a (probability) measure on $(M,\mathcal{M})$, $\mu_X:= P\circ X^{-1}$, which is called \emph{probability distribution of the random variable $X$}. Statistical information about a random variable $X$ can be obtained from the \emph{expectation value}, defined as
	\begin{equation*}
		\Ex[X] := \int_{\Omega} X(\omega) P(d\omega)
	\end{equation*}
	Sometimes, to emphasise the probability measure which we are using to compute the expectation we write $\Ex_P$. 
	
	The definition given for $P$ is rather obscure since we should explain the meaning of \textquotedblleft degree of belief". This is a signature of the fact that the notion of probability is a primitive concept in the measure-theoretic formulation. A method we can use to measure $P$ is explained in the \emph{(weak) law of large numbers}.
	\begin{theorem}
		Let $\{A_i\}_{i \in \Nat}$ be a collection of \emph{independent events}, namely $P(\cap_i A_i) = \Pi_i P(A_i)$. If $P(A_i) = p$ for all $i$, namely they have all the same probability, for any $\epsilon > 0$, we have
		\begin{equation*}
			\lim_{n \rightarrow \infty} P\left( \bigg\{ \omega \in \Omega \bigg| \bigg|\frac{K_n(\omega)}{n} - p\bigg| < \epsilon \bigg\} \right) = 1
		\end{equation*}
		where $K_n(\omega) = \sum_{j=1}^n \chi_{A_j}(\omega)$ where $\chi_{A_j}(\omega)$ is the indicator function for the set $A_j$.
	\end{theorem} 
	Let us explain the meaning of this theorem, and what it tells us about the measurement of probability. First of all, we should accept that $P(A) =1$ means to be sure that the event $A$ is true. Assumed this, the meaning of this theorem is hidden in the function $K_n(\omega)$. Consider the following collection of events
	\begin{equation*}
		A_j := \{ \mbox{in the trial $j$ we found } \omega \in A\}  \qquad \forall j \in \Nat.
	\end{equation*}
	For this collection of events, the function $K_n(\omega)$ is just the number of times we observe $\omega \in A$ by repeating the observation $n$ times. In addition the independence hypothesis in the theorem ensures that the observation in the $i$-th trial does not influence the $j$-th trial. Finally the requirement $P(A_i) = p$, $\forall i \in \Nat$, is a quite natural requirement: the probability that $\omega \in A$ is the same independently on the trial. At this point the meaning of the theorem is clear: for sufficiently many trials, the number of times we find $\omega \in A$ normalised to the total number of trials, tend to be the number $p$ with probability $1$. Notice that, despite this theorem tells us how to \emph{measure} $P$ (via frequencies), we cannot use it to define the meaning of $P$: it would be a recursive definition. For this reason, probability in the measure-theoretic framework  is a primitive notion.
	
	\paragraph{Example: The die}
	From the mathematical point of view, a die can be described using the probability space $\PS$. For example let us set
	\begin{enumerate}
		\item[i)] $\Omega := \{1,2,3,4,5,6\}$; 
		\item[ii)] $\E$ is the power set of $\Omega$;
		\item[iii)] $P(\omega) = \sum_{i=1}^{6} p_i \delta_{i,\omega}$ with $p_i \in [0,1]$ for all $i$ and $\sum_{i=1}^{6} p_i =1$.
	\end{enumerate}
	The die in this case is represented by a random variable $D:\Omega \rightarrow \Nat$, which is just the identity function, $D(\omega) := \omega$ and the distribution of the random variable is given by the probability measure directly. Another possibility can be the following
	\begin{enumerate}
		\item[i)] $\Omega := \{1,2,3,4,5,6,7,8\}$;
		\item[ii)] $\E$ is the power set of $\Omega$;
		\item[iii)] $P(\omega) = \sum_{i=1}^{8} q_i \delta_{i,\omega}$ with $q_i \in [0,1]$ for all $i$ and $\sum_{i=1}^{8} q_i =1$.
	\end{enumerate}
	On this probability space  the die can be described using the random variable
	\begin{equation*}
		D(\omega) := 
		\begin{cases}
			\omega \mbox{ if } \omega \in \{1,3,4,6\} \\
			2 \mbox{ if } \omega \in \{5,8\} \\
			5 \mbox{ if } \omega \in \{2,7\}
		\end{cases}
	\end{equation*}
	Again $D(\Omega) = \{1,2,3,4,5,6\}$, but this time the probability distribution of $D$ does not coincide with the probability measure, in particular we have that
	$\mu_D = \{q_1,q_5+q_8,q_3,q_4,q_2 + q_7,q_6\}$.

	\subsection{The algebra of functions $L_{\infty}\PS$}\label{sec1b}
	
	Let $\PS$ be a probability space and consider the functions $f: \Omega \rightarrow \Comp$ which are measurable with respect to the $\sigma$-algebra $\E$. We also require $f$ to be bounded, namely $\|f\| := \sup_{\omega \in \Omega} |f(\omega)| < \infty$.  Then we can define the following class of functions
	\begin{equation*}
		\mathscr{L}_\infty \PS := \{f:\Omega \rightarrow \Comp | \mbox{$f$ is $\E$-measurable}, \|f\| < \infty\}
	\end{equation*}
	With the equivalence relation $f \sim g$ whenever $P(\{\omega \in \Omega | f(\omega) = g(\omega)\}) = 1$, we can define the following object $L_{\infty} \PS = \mathscr{L}_\infty \PS / \sim$. Defining the operation of sum, multiplication by a scalar and multiplication between functions in the usual way, $L_{\infty} \PS$ becomes an algebra of functions. Finally, using the \emph{essential supremum norm} $\|f\|_\infty := \mbox{ess sup} f = \inf_{\alpha \in \Rea}P\left( \{\omega \in \Omega | |f(\omega)| < \alpha\} \right)$, $L_{\infty} \PS$ is an abelian $C^*$-algebra of functions ($C^*$ means that $\|f^*f\|_{\infty} = \|f\|^2_{\infty}$ which is true for the complex conjugation $^*$ and the essential supremum norm $\| \cdot \|_\infty$, see section \ref{APdef} for more details). Note that $\mbox{ess sup} f \leqslant \sup f$.
	
	This object encodes, in an algebraic way, all the information encoded in the underlying probability space. Clearly $\PS$ determines uniquely $L_\infty \PS $, but the opposite is not exactly true. Indeed, given $L_\infty \PS$ we may construct a $\sigma$-algebra by setting
	\begin{equation*}
		\tilde{\E} := \{ p \in L_\infty \PS | p=p^* = p^2 \}
	\end{equation*}
	but this is not isomorphic to the original $\E$, since we identified everywhere $P$-equal functions in the construction of $L_\infty \PS$. $\tilde{\E}$ is a \emph{measure algebra}. Nevertheless this is an advantage instead of a limitation. Indeed, measure algebra is a coherent way to exclude set of zero measure from the probability space describing the random phenomenon (see Sec. 1.7 in \cite{petersen1989ergodic}). Over this $\sigma$-algebra, we can define a probability measure $\tilde{P}: \tilde{\E} \rightarrow [0,1]$ as $\tilde{P}(f) := \phi(f)$ for any $f$ which is $\tilde{\E}$-measurable, where $\phi$ is a positive normalised linear functional defined to be $\phi(f) = \int_{\Omega} f(\omega) P(d\omega)$. Summarising, starting from the algebra $L_\infty \PS$ we can construct a probability space $(\Omega,\tilde{\E},\tilde{P})$ which is equivalent, up to zero measure set, to the probability space $(\Omega,\E,P)$.
	
	Consider an ordinary random variable $X:\PS \rightarrow (M,\mathcal{M})$. Since $(M,\mathcal{M},\mu_X)$ is a probability space as well, we can associate to it an abelian $C^*$-algebra of functions. In this picture, $X$ can be seen as a linear map between algebras which respects the multiplication, namely a $C^*$-\emph{algebra homomorphism}. Thus we can say that the algebraic analogous of the random variable $X$ is the $C^*$-algebra homomorphism $x_X: L_\infty \PS \rightarrow L_\infty (M,\mathcal{M},\mu_X)$. We can see that a random phenomenon described in measure theoretic language, can be equivalently described using (abelian) algebras: this is part of the algebraic approach to probability theory. We conclude this section by observing that, in the algebraic approach, the role of the random variables is central: the elements of $L_\infty \PS$ are functions on $\PS$, i.e random variables.
	
    \subsection{From probability spaces to abelian von Neumann algebras}\label{APdef}\label{sec1c}
    
    We have seen that the information encoded in $\PS$, can be encoded in an equivalent manner in the algebra $L_\infty \PS$. Now, we will establish a link between the algebra $L_\infty$ and a suitable von Neumann algebra of operators over some Hilbert space.
    
    In general, an algebra $\A$ is a vector space equipped with a product operation. Typically such product is assumed to be associative and in some case, it can be commutative.
    An algebra can have or not have the unit element with respect to this multiplication but in what follows we will always consider algebras with unit. We will always consider algebras having a \emph{norm} defined on it, labeled by $\| \cdot \|$. It is also useful to consider algebras equipped with an additional map $^*:\A \rightarrow \A$, such that $(a^*)^* = a$, which is called \emph{involution} (examples of involutions are the complex conjugation for functions or the adjoint operation for operators). At this point, we may define what is a $C^*$-algebra.
    \begin{definition}
    	Let $\A$ be an algebra with a norm $\| \cdot \|$ and an involution $^*$. If $\A$ is complete with respect to the norm $\|\cdot \|$ we call this algebra \emph{$^*$-algebra}.  If in addition,    \begin{equation*}
    		\| a^* a \| = \|a \|^2
    	\end{equation*}
    	we say that $\A$ is a \emph{$C^*$-algebra}.
    \end{definition}
    Completeness of $\A$ is understood in the usual way: all the convergent sequences in $\A$ with respect to a given norm are also Cauchy sequences.
    Given a $^*$-algebra $\A$, a generic element $a \in \A$ is said to be \emph{self-adjoint} if $a = a^*$, while it is said to be positive (and we will write $a \geqslant 0$) if we can write $a = b^*b$, for some $b\in A$.
    \begin{definition}\label{State-def}
    	Let $\A$ be a $^*$-algebra , a \emph{state over }$\A$ is a linear functional $\phi:\A \rightarrow \Comp$ which is positive ($\phi(aa^*) \geqslant 0$ for any $a \in \A$) and normalised ($\phi(\unit) = 1$, where $\unit$ is the unit of $\A$).
    \end{definition}
    Note that the definitions above are very abstract in the sense that we do not need to define explicitly the sum, the product, the norm or the involution. For this reason $\A$ is called \emph{abstract algebra} if such information are not declared. When all the features of the algebra are explicited, we speak of \emph{concrete algebra}. Let us restrict our attention to the case of algebras of operators in some Hilbert space $\Hi$, and in particular to $\A =\boundo$ (the bounded operators over an Hilbert space $\Hi$) which is a concrete algebra. Thanks to the notion of positivity, we have a natural ordering operation $\geqslant$ between the elements of the algebra, i.e. given two operators $\hat{A}_1$ and $\hat{A}_2$, the writing $\hat{A_1} \geqslant \hat{A_2}$ means $\hat{A}_1 - \hat{A}_2 \geqslant 0$. 
    \begin{definition}
    	Let $\A$ be an operator algebra and $\{\hat{A}_i\}_{i=1,2,\cdots}$ be an increasing sequence of operators in $\A$ \emph{with strong limit} $s-\lim_{n \rightarrow \infty} \hat{A}_n =\hat{A}$, namely $\hat{A}_1\leqslant \hat{A}_2 \leqslant \cdots$ and $\lim_{n \rightarrow \infty} \|\hat{A}_n - \hat{A}\| = 0$ for some $\hat{A} \in \A$. A state $\phi$ is said to be \emph{normal} if $\lim_{n \rightarrow \infty} \phi(\hat{A}_n) = \phi(\hat{A})$.
    \end{definition}    
    A normal state $\phi$ on $\boundo$ can be written as $\phi(\cdot) = \Tr{\hat{\rho} \mbox{ }\cdot\mbox{ }}$ for some $\hat{\rho} \in \tco$, where $\tco$ labels the set of trace-class operators over $\Hi$ (see Th. 7.1.12 in \cite{kadison2015fundamentals}).
    \begin{definition}
    	Let $\A$ be an algebra of operators and $\phi:\A \rightarrow \Comp$ a state on it. Take some $\hat{A}\in \A$, if $\phi(\hat{A}^*\hat{A}) = 0$ implies $\hat{A} = 0$, then $\phi$ is said \emph{faithful}.
    \end{definition}
    Among algebras of operators a very important class is the one of von Neumann algebras.
    \begin{definition}
    	Let $\Hi$ be an Hilbert space, a \emph{von Neumann algebra} $\vona$ is a $^*$-sub-algebra of $\boundo$ which is strongly closed (i.e. the strong limit of any sequence of operators in $\vona$ converge to some operator which is still in $\vona$).
    \end{definition}
    In general, any von Neumann algebra is a $C^*$-algebra, but the opposite is not true. Von Neumann algebras are concrete algebras, however in general one should consider more abstract algebras, not necessarily composed of operators, hence it is useful to introduce also the notion of representation.
    \begin{definition}
    	Let $\A$ be an algebra with involution and $\Hi$ an Hilbert space. An homomorphism $\pi: \A \rightarrow \boundo$ preserving the involution is called \emph{representation of $\A$ on $\Hi$}. A representation is said \emph{faithful} if it is one-to-one.
    \end{definition}
    
    We now have all the notions needed to state the main theorem of this section. Consider the algebra $L_\infty \PS$ and for any $f \in L_\infty \PS$ define the operator $\hat{M}_f$ on the Hilbert space $L_2\PS$ as
    \begin{equation*}
    	\hat{M}_f \psi(\omega) = f(\omega) \psi(\omega) \mspace{30mu} \psi(\omega) \in L_2(\Omega,\E,P)
    \end{equation*}
    Clearly, such a representation is faithful and represents $L_\infty \PS$ as multiplicative operators on $L_2 \PS$. This is the link mentioned in the beginning. More formally we have the following theorem.
    \begin{theorem}\label{theo2}
    	Let $\PS$ be a probability space. Then the algebra $\mathcal{V}_c(L_2\PS):= \{\hat{M}_f | f \in L_\infty\PS \}$ is an abelian von Neumann algebra on the Hilbert space $L_2\PS$ and
    	\begin{equation*}
    		\phi_P: \hat{M}_f \mapsto \int_{\Omega} f(\omega) P(d\omega)
    	\end{equation*}
    	is a faithful normal state on $\mathcal{V}_c(L_2\PS)$.
    \end{theorem}
    \begin{proof}
    	See Appendix A.
    \end{proof}
    More generally, the results obtained till now can be reversed: starting from a generic abelian von Neumann algebra we may construct a probability space \cite{massen1998quantum}.
    \begin{theorem}\label{th111}
    	Let $\A$ be an abelian von Neumann algebra of operators and $\phi$ a faithful normal state on it. Then there exist a probability space $\PS$ and a linear correspondence between $\A$ and $L_\infty \PS$, $\hat{A}\mapsto f_{\hat{A}}$, such that
    	\begin{equation*}
    		\begin{split}
    			f_{\hat{A}\hat{B}} = f_{\hat{A}} f_{\hat{B}} \mspace{50mu} f_{\hat{A}^*} = (f_{\hat{A}})^* \\
    			\| f_{\hat{A}} \|_\infty = \| \hat{A} \| \mspace{50mu} \Ex[f_{\hat{A}}] = \phi(\hat{A}).
    		\end{split}
    	\end{equation*}
    \end{theorem}
    Summarising, the theorem above tells that any abelian $C^*$-algebra of functions, which is constructed from a probability space, can be described in an equivalent way by using multiplicative operators over a suitable Hilbert space that one can construct from the probability space itself. It is important to observe that the state $\phi_P$ is not constructed from the vectors of $L_2\PS$. Finally, despite we are describing a classical probability space using an Hilbert space, this Hilbert space changes when we change the probability measure $P$.
	
	\subsection{Essentials of spectral theory for bounded operators}\label{sec1d}
	
	Here we introduce the basic notions and theorems about the spectral theory of bounded operators which will be used later \cite{moretti2013spectral}. The central object of the spectral theory is the notion of PVM. In order to define them in the whole generality, we recall that a \emph{second-countable topological space} $\topo$ is a set $\X$ with a topology $\mathscr {T}$ (collection of open sets) whose elements can be seen as the countable union of basis sets (i.e. elements of $\mathscr {T}$  which cannot be seen as unions of other sets).
	\begin{definition}\label{PVMdef}
		Let $\Hi$ be an Hilbert space, $\topo$ a second-countable topological space and $\borelx$ the borel $\sigma$-algebra on $\X$. The map $\hat{P}: \borelx \rightarrow \boundo$ is called \emph{projector-valued measure} (PVM) on $\X$, if the following conditions holds
		\begin{enumerate}
			\item[i)] $\hat{P}(B) \geqslant 0$ for any $B \in \borelx$;
			\item[ii)] $\hat{P}(B)\hat{P}(B') = \hat{P}(B \cap B')$ for any $B,B' \in \borelx$;
			\item[iii)] $\hat{P}(\X) = \Id$;
			\item[iv)] if $\{B_n\}_{n \in \Nat} \subset \borelx$ with $B_n \cap B_m = \{\varnothing\}$ for $n \neq m$, then
			\begin{equation*}
				\sum_{n = 0}^{\infty} \hat{P}(B_n) = \hat{P} \left( \bigcup_{n = 0}^\infty B_n\right)
			\end{equation*}
		\end{enumerate}
		The support of the PVM is the closed set defined as $\mbox{supp}(\hat{P}) := \X/ \{\cup A |A \in \mathscr {T}, \hat{P}(A) = \hat{\mathbb{O}}\}$. When $\X = \Rea^n$, $\hat{P}$ is said bounded if $\mbox{supp}(\hat{P})$ is a bounded set.
	\end{definition}
	In the above definition $\hat{\mathbb{O}}$ is simply the null operator. Because a $(\Int, \mathcal{P}(\Int) )$, where $\mathcal{P}(A)$ means the power set of $A$, and $(\Rea^n, \mathcal{T}_o)$, where $\mathcal{T}_o$ is the ordinary euclidean topology, are second-countable topological spaces, with the above definition we may treat at the same time the continuous and discrete cases. PVMs are useful because they allow to define operator-valued integrals with respect to them. In fact, if we consider a bounded function $g:\X \rightarrow\Comp$ which is measurable, we can define
	\begin{equation*}
		\hat{F}(g) := \int_{\X} g(x)\hat{P}(dx)
	\end{equation*}
	which is called \emph{integral operator in $\hat{P}$} and it is a (bounded) operator on $\Hi$. We observe that 
	\begin{equation*}
		\int_{\X} g(x)\hat{P}(dx) =\int_{\mbox{supp}(\hat{P})} g(x)\hat{P}(dx)
	\end{equation*}
	for any measurable bounded function $g$, because the PVM vanishes for all $A \notin \mbox{supp}(\hat{P})$. One can prove that, if $g$ is measurable and bounded, then $\| \int_{\X} g(x) \hat{P}(dx) \| \leqslant \| g|_{\mbox{supp}(\hat{P})} \|_\infty$ and also that the integral operator is positive for $g$ positive. Related to PVM, another important quantity is the \emph{spectral measure}.
	\begin{definition}
		Let $\psi \in \Hi$, the map $\mu_{\psi}:\borelx \rightarrow \Rea$ defined as
		\begin{equation*}
			\mu_{\psi}(E) := \langle \psi | \int_{\X} \chi_E(x)\hat{P}(dx) \psi \rangle \mspace{30mu} E \in \borelx,
		\end{equation*}
		where $\langle \cdot | \cdot \rangle$ is the scalar product of $\Hi$, is a real and positive measure on $\Rea$ called \emph{spectral measure associated to} $\psi$.
	\end{definition}
	Note that if $\psi \in \Hi$ is normalised, then also $\mu_{\psi}$ is. An important property for any bounded and measurable function $g$ on $\X$ is the following:
	\begin{equation*}
		\langle \psi | \int_{\X} g(x) \hat{P}(dx) \psi \rangle = \int_{\X} g(x) \mu_{\psi}(dx)
	\end{equation*}
	which is simply a consequence of the fact that we may always write $g(x)$ as limit of a sum of indicator functions. This last equality is very important in quantum physics. Since $\int_{\X} f(x) \hat{P}(dx)$ is an operator on $\Hi$, i.e. $\hat{F}(g)$, the above equality tells that the quantum mechanical expectation $\langle \psi | \hat{F}(g) \psi \rangle$ coincides with the ordinary expectation value $\Ex[g]$ when it is computed with the spectral measure. As we will see, this is a very general feature of self-adjoint operators.
	
	At this point we can state (without proof) the two central theorems of the spectral theory for bounded operators. The first important theorem is the \emph{spectral decomposition theorem for self-adjoint operators in $\boundo$} which tells that every self-adjoint operator in $\boundo$ can be constructed integrating some function with respect to a specific PVM, and it is completely determined by it. 
	\begin{theorem}[Th. 8.54 in \cite{moretti2013spectral}]
		Let $\Hi$ be an Hilbert space and $\hat{A} \in \boundo$ a self-adjoint operator. 
		\begin{enumerate}
			\item[a)] There exists a unique and bounded PVM $\hat{P}^{(\hat{A})}$ on $\Rea$ such that
			\begin{equation*}
				\hat{A} = \int_{\mbox{supp}(\hat{P}^{(\hat{A})})} x  \hat{P}^{(\hat{A})}(dx);
			\end{equation*}
			\item[b)] $\sigma(\hat{A}) = \mbox{supp}(\hat{P}^{(\hat{A})})$, where $\sigma(\hat{A})$ is the spectrum of the operator $\hat{A}$;
			\item[c)] If $f$ is a bounded measurable function on $\sigma(\hat{A})$, the operator $f(\hat{A}) := \int_{\sigma(\hat{A})} f(x)\hat{P}^{(\hat{A})}(dx)$ commutes with every operator in $\boundo$ which commutes with $\hat{A}$.
		\end{enumerate}
	\end{theorem}
	The second important theorem is the so called \emph{spectral representation theorem of self-adjoint operators in $\boundo$}. This theorem tells that every bounded self-adjoint operator on $\Hi$ can be represented as a multiplicative operator on some $L_2$ Hilbert space, which is basically constructed from its spectrum.
	\begin{theorem}[Th. 8.56 in \cite{moretti2013spectral}]\label{Spectheo1}
		Let $\Hi$ be an Hilbert space, $\hat{A} \in \boundo$ a self-adjoint operator and $\hat{P}^{(\hat{A})}$ the associated PVM. Then
		\begin{enumerate}
			\item[a)] $\Hi$ splits as Hilbert sum $\Hi = \oplus_{i \in I}\Hi_i$ (with $I$ at most countable if $\Hi$ is separable), where $\Hi_i$ are closed and mutually orthogonal subspaces such that
			\begin{enumerate}
				\item[i)] $\forall i \in I$, then $\hat{A}\Hi_i \subset \Hi_i$;
				\item[ii)] $\forall i \in I$ there exist a positive finite borel measure $\mu_i$ on the Borel sets of $\sigma(\hat{A})\subset\Rea$, and a surjective isometry $\hat{U}_i:\Hi_i \rightarrow L_2(\sigma(\hat{T}),\mu_i)$ such that
				\begin{equation*}
					\hat{U}_i \left( \int_{\sigma(\hat{A})}f(x)\hat{P}^{(\hat{A})}(dx)\right)\bigg|_{\Hi_i} \hat{U}_i^{-1} = f \cdot
				\end{equation*}
				for any bounded measurable $f$, where $f\cdot$ means multiplication by $f$ in $L_2(\sigma(\hat{A}),\mu_i)$.
			\end{enumerate}
			\item[b)] $\sigma(\hat{A}) = \mbox{supp}(\{\mu_i\}_{i \in I})$ where $\mbox{supp}(\{\mu_i\}_{i \in I})$ is the complement to the set of $\lambda \in \Rea$  for which there is an open set $B_\lambda \subset \Rea$ such that $\lambda \in B_\lambda$ and $\mu_i(A_\lambda) = 0$ for all $i \in I$.
			\item[c)] If $\Hi$ is separable there exist a measure space $(M_A,\Sigma_A,\mu_A)$ with $\mu_A(M_A) < +\infty$, a bounded map $F_A: M_A\rightarrow\Rea$ and a unitary operator $\hat{U}_A:\Hi \rightarrow L_2(M_A,\mu_A)$ satisfying
			\begin{equation*}
				\left(\hat{U}_A \hat{A} \hat{U}^{-1}_A g\right)(x) = F_A(x)g(x)
			\end{equation*}
			for any $g \in \Hi$.
		\end{enumerate}
	\end{theorem}
	Note that in $c)$ the measure is not uniquely determined by $\hat{A}$. This theorem is a more general version of the well known result about the splitting of an Hilbert space as direct sum of eigenspaces associated to a self adjoint operator.
	
	Let us conclude this section observing that the spectral decomposition theorem tells that any self-adjoint operator (i.e. a possible quantum observable) can always be seen as an integral operator and that this decomposition is unique. The spectral measure allows to compute the quantum expectation as an ordinary expectation and, finally, the spectral representation theorem tells that the whole algebraic structure described in section \ref{APdef} is present. This means that the complete probabilistic description of a \emph{single quantum observable} is possible by using measure-theoretic probability.
	
	\subsection{Ordinary probability in Hilbert spaces} \label{PinH}\label{sec1e}
	
	We concluded the previous section observing that for a single quantum observable we can use measure-theoretic probability without problems. In this section we want to see how we can do the opposite: describe a measure-theoretic random variable with operators over an Hilbert space. In section \ref{sec1b} we have seen that to any measure-theoretic probability space, $\PS$, we may associate an abelian $C^*$-algebra of functions, $L_\infty \PS$, which can always be represented by using multiplicative operators over the Hilbert space $L_2 \PS$, i.e. the commutative von Neumann algebra $\mathcal{V}_c(L_2\PS)$, as shown in theorem \ref{theo2}. Such a theorem also tells that expectations with respect to a probability measure are nothing but states over $\mathcal{V}_c(L_2\PS)$. We also observed that the Hilbert space $L_2 \PS$ strongly depends on the probability measure of the underlying probability space, and so a change of the probability measure would change the Hilbert space. However, the spectral representation theorem suggests that we may find a \textquotedblleft bigger Hilbert space" (namely $\Hi = \oplus_i \Hi_i$, as defined in the theorem) where this dependence on the probability measure seems to disappear. Finally, the spectral measure, introduced in section \ref{sec1d}, seems to allow us to move the probabilistic content from the original probability measure to (functional of) function of this \textquotedblleft bigger Hilbert space". In this section we want to study better this mechanism. More precisely, we want to discuss the following problem: how it is possible to construct explicitly an Hilbert space (independent on the probability measure), an operator and a \emph{state} (defined as in definition \ref{State-def}) on a suitable algebra of operators on $\Hi$, which are capable to give the same statistical prediction about a random variable described in ordinary measure-theoretic setting.
	
	Consider a probability space $\PS$, a measurable space $(M,\mathcal{M})$ and a random variable $X:\Omega \rightarrow M$ on it. As usual $X$ induces a distribution $\nu_X$ such that $(M,\mathcal{M}, \nu_X)$ is a probability space. Algebraically, the random variable $X$ can be seen as the map $x:L_\infty \PS \rightarrow L_{\infty}(M,\mathcal{M},\nu_X)$. Clearly the random variable $X$ can be seen also as the identity map on $L_{\infty}(M,\mathcal{M},\nu_X)$, and any expectation $\Ex_P$ can be computed using a suitable state $\phi_{\nu_X}$ over this algebra, i.e. $\Ex_{P}[f(X)] = \phi_{\nu_X}(f(X))$. This fact does not change if we represent the element $x \in L_{\infty}(M,\mathcal{M},\nu_X)$, corresponding to the original random variable $X$, as a multiplicative operator $\hat{M}_x$ acting on $L_2 (M,\mathcal{M},\nu_X)$, i.e. if we consider the abelian von Neumann algebra of operators $\A := \{\hat{M}_f | f \in L_2 (M,\mathcal{M},\nu_X)\}$ on this Hilbert space. Clearly $L_2 (M,\mathcal{M},\nu_X)$ changes as we change the initial probability measure $P$. Consider now the Hilbert space $\Hi = \oplus_i \Hi_i$ of the spectral representation theorem and a bounded operator $\hat{T} \in \boundo$ on it with spectrum $\sigma(\hat{T})$. Then take the surjective isometry of the theorem, i.e. $\hat{U}_i: \Hi_i  \rightarrow L_2(\sigma(\hat{T}),\mu_i)$. The idea is to use $\hat{U}_i$ to map $L_2(M,\mathcal{M}, \nu_X)$ in some $\Hi_i$ and to construct $\Hi$ from it. If we want to do that we can set:
	\begin{enumerate}
		\item[a)] $\sigma(\hat{T}) = M$,
		\item[b)] $\mu_i = \nu_X$.
	\end{enumerate}
	This allows to write that $\hat{U}_i: \Hi_i  \rightarrow L_2(M,\nu_X)$ (we omit the $\sigma$-algebra $\mathcal{M}$ for simplicity). 
	These requirements can be explained as follows. Since we want to represent with $\hat{T}$ the random variable $X$ (note that $\hat{T}$ is not the operator $\hat{M}_x$ seen before) and encode the probabilistic content of $(M,\mathcal{M}, \nu_X)$ (and so of $\PS$) in some suitable object defined on $\Hi$, the requirement $a)$ simply  means that the set of eigenvalues of the operator coincides with the set of outcomes of the random variable. This tells us how to construct the operator $\hat{T}$ since the spectrum uniquely identifies the operator. The requirement $b)$ is needed in order to encode the statistical information in functionals of elements of $\Hi$, allowing the Hilbert space, on which $\hat{T}$ is defined, to be capable to contain information about $P$. Note that at this level it is not clear what the meaning of the index $i$ is (which is important for the construction of $\Hi = \oplus_i \Hi_i$) in the original probability space. Observing that this index determines the dimension and separability property of the Hilbert space, let us try to attach it to some feature of the random variable we want to represent. In particular, \emph{ we assume that $i$ labels the outcome of $X$, i.e. $i = x \in M$}. This immediately implies that
	\begin{equation*}
		\Hi = \bigoplus_{x\in M} \Hi_x
	\end{equation*}
	where $\oplus$ means direct sum or direct integral according to the cardinality of $M$, while the operator representing $X$ is simply
	\begin{equation*}
		\hat{T} := \int_{M} x \hat{P}^{(\hat{T})}(dx)
	\end{equation*}
	where $\hat{P}^{(\hat{T})}(dx)$ is the PVM having $M$ as support. Note that in this way $\hat{T}\Hi_x \subset \Hi_x$, i.e. $\hat{P}^{(\hat{T})}(dx) | x \rangle = | x \rangle$ for any $| x \rangle \in \Hi_x$, as required by the spectral decomposition theorem. By construction the operator $\hat{T}$ has a non-degenerate spectrum and if $M$ is a bounded subset of $\Rea$, the spectrum of $\hat{T}$ is bounded, implying that $\hat{T}$ is a bounded operator. Let us assume this for the rest of this section. The only thing that we miss is how to represent the probability distribution $\nu_X$. At this point, we assume that the random variable $X$ is discrete hence $\nu(x)$ can be interpreted as the probability to have $X = x$. In general on $(M,\mathcal{M},\nu_X)$ we can describe, together with  $X$, all random variables $f(X)$, where $f:M \rightarrow M$ are measurable and bounded functions, and they correspond to the operators $f(\hat{T})$. Hence, since $\nu(x)$ is a bounded and measurable function, it can be represented as
	\begin{equation*}
		\hat{\rho}_{\nu}:= \nu(\hat{T}) = \int_{M} \nu(x) \hat{P}^{(\hat{T})}(dx).
	\end{equation*}
	Note that $\hat{\rho} \in \tco$, because $\nu(x)$ is a probability. In section \ref{sec1c} we have seen that a normal state $\phi(\cdot)$ on $\boundo$ can be always written as $\Tr{\hat{\rho} \mbox{ }\cdot\mbox{ }}$ for some trace class operator $\hat{\rho}$. The set of all the operators $f(\hat{T})$, equipped with the operation of sum and product of operators, forms a sub-algebra of $\boundo$ which is in one-to-one correspondence (via the surjective isometry $\hat{U}_x$) with an abelian von Neumann algebra. Thus this set of operators form an abelian von Neumann algebra, which we label by $\mathcal{V}_T$. Then if we impose that states on $\mathcal{V}_T$ coincide with states of $L_{\infty}(M,\mathcal{M},\nu_X)$ (inheriting all their properties), we must have
	\begin{equation}\label{QuantumExp}
		\Ex_{\nu_X}[f(X)] = \Tr{\hat{\rho} f(\hat{T})}
	\end{equation}
	for any $f$ measurable and bounded function on $M$. This implies that $\hat{\rho} = \hat{\rho}_\nu$. Note that, this time  given $\nu_X$ (i.e. $P$) we can determine a unique object which encodes all the probabilistic information of the random phenomenon under study.\newline
	
	Heuristically, it seems that we can write the following formal \textquotedblleft correspondence"
	\begin{equation*}
	\begin{split}
	P(d\omega) &\leftrightarrow \hat{\rho}_P \\
	\int \cdots  &\leftrightarrow \Tr{\cdots}
	\end{split}
	\end{equation*} 
	which anyhow should be taken with care. First, additional difficulties are added if one drops the assumption that $M$ is a bounded subset of $\Rea$. Another difficulty arises if we want to describe continuous random variable taking value on $\Rea$. These difficulties may be overcome, from a practical point of view, by seeing continuous unbounded operators as the limit of bounded operators with discrete spectrum: this is the solution that we will adopt in \cite{LC2} and \cite{LC3} to deal with continuous unbounded random variables. Rigorous approaches to treat algebraically unbounded operators are available \cite{kadison2015fundamentals} while the notion of (generalized) eigenvalues for continuous unbounded operators can be formalized, from the mathematical point of view, using the \emph{Gel'fand triples} \cite{bohm1989dirac}. Despite seems to be an overcomplication, this change of language for the description of a random phenomenon gives rise to new possibilities, as it will be explained in the next sections.\newline
	
	\paragraph{Example: The die}
	Let us continue the previous example of the classical die. Consider the first description we gave in the previous example, i.e. we used the probability space $(\Omega,\E,P)$ with $\Omega = \{1,\cdots, 6\}$. If $\hat{D}$ is the operator associated to the random variable $D: \Omega \rightarrow \Nat$, we know that
	\begin{equation*}
		\sigma(\hat{D}) = \Omega = \{1,2,3,4,5,6\}.
	\end{equation*}
	The Hilbert space on which this operator act is $\Hi = \oplus_{i \in \Omega} \Hi_i$. It has dimension $6$ and in general it can be seen as a subspace of $\Comp^6$. A generic random variable $X(\omega) = f(\omega)$ over $(\Omega,\E,P)$ corresponds to the operator
	\begin{equation*}
		f(\hat{D}) = \sum_{i=1}^{6} f(i) |i\rangle\langle i |.
	\end{equation*} 
	The probability measure, can be represented as
	\begin{equation*}
		\hat{\rho}_P = \sum_{i=1}^{6} p_i |i\rangle\langle i|,
	\end{equation*}
	where $|i\rangle\langle i|$ is the projector on $\Hi_i$. Thus, any expectation value can be computed as
	\begin{equation*}
		\Ex_P[X] = \Tr{ f(\hat{D})\hat{\rho}_P}.
	\end{equation*}
	We again stress that we are using just one basis of $\Comp^6$, so \emph{operators written in different basis do not correspond to any random variable which can be defined on the original probability space} and for this reason they must be excluded (they do not belong to the same abelian algebra of $\hat{D}$).
	
    \section{Algebraic probability spaces}\label{section2}
    
    In part \ref{SECT1}, we have seen that the usual measure-theoretic formulation of probability theory can be encoded in a satisfactory way in an abelian von Neumann algebra of functions. This suggests that a more general formulation of probability theory is possible in an algebraic context, allowing to obtain a non-commutative probability theory. Here we will present the basic facts about algebraic probability theory, emphasizing the role of commutativity and its influence on the possible concrete representations of such algebraic spaces. Additional references for that section are \cite{redei2007quantum,accardiprobabilita,voiculescu1992free}.
    
    \subsection{Basic definitions}
    
    Some of the basic definitions we need in order to describe the algebraic approach to probability have been already introduced in section \ref{APdef}. For notions like algebra, involution, state (and its classification) and representation, we will refer to this section. 
    \begin{definition}
    	The pair $(\A,\omega)$ where $\A$ is a $^*$-algebra with unit, and $\omega:\A \rightarrow \Comp$ a state on it, is called \emph{algebraic probability space}.
    \end{definition}
    We will restrict our attention to the case where $\A$ is  also $C^*$ quickly. Note that the commutativity of $\A$ is not required in the definition. In section \ref{sec1c} we saw that in the abelian case, if $a\ \in \A$, $\omega(a) \in \Comp$ is its expectation value. If $a = a^*$ (i.e. $a$ is self-adjoint), then one can prove that $\omega(a) \in \Rea$, thus self-adjoint elements of the algebra correspond to real-valued random variables. More generally, the elements of a generic algebra can be interpreted as random variable, as the following definition suggests.
    \begin{definition}
    	Given an algebraic probability space $(\A,\omega)$ and another $^*$-algebra with unit $\B$, then an homomorphism $j:\B\rightarrow\A$ preserving the unit and the involution, is called an \emph{algebraic random variable}.
    \end{definition}
    This definition is just an extension of the notion of random variable, used in the abelian case, to general algebras. As in ordinary probability theory, the algebraic random variable $j$, induces a state, $\omega_j= \omega \circ j$  called \emph{distribution}, such that $(\B,\omega_j)$ is another algebraic probability space.
    
    \subsection{Representations of an algebra}
    
    Algebras are very abstract objects. For this reason, the notion of representation is very important. Here we will review the two basic representation theorems that we have at disposal, in order to pass from an abstract algebra to some concrete algebra.
    
    A general result which allows to represent a generic abstract $C^*$-algebra with a concrete $C^*$-algebra of operators is the celebrated GNS theorem. First we need to introduce some terminology: a representation is called a \emph{$^*$-representation} if it preserves the involution, while a vector $\psi \in \Hi$ is said to be \emph{cyclic} for a representation $\pi$, if $\mbox{span} \{\pi(a)\psi | a \in \A\}$ is a dense subspace of $\Hi$.
    \begin{theorem}
    	Let $\A$ be a $C^*$-algebra with unit and $\omega:\A\rightarrow\Comp$ a state. Then:
    	\begin{enumerate}
    		\item[i)] there exists a triple $(\Ho,\pi_\omega,\Psi_\omega)$ where $\Ho$ is an Hilbert space, $\pi_\omega : \A \rightarrow 
    		\boundoo$ is a $^*$-representation  of $\A$ on the $C^*$-algebra of bounded operators on $\Ho$, and 
    		$\Psi_\omega \in \Ho$ is a vector, such that:
    		\begin{enumerate}
    			\item[a)] $\Psi_\omega$ is a unit vector, cyclic for $\pi_\omega$;
    			\item[b)] $\langle \Psi_\omega | \pi_\omega(a) \Psi_\omega \rangle = \omega(a)$ for any $a \in \A$.
    		\end{enumerate}
    		\item[ii)] If $(\mathcal{H},\pi,\Psi)$ is a triple such that:
    		\begin{enumerate}
    			\item[a)] $\mathcal{H}$ is an Hilbert space, $\pi: \A \rightarrow \boundo$ is a $^*$-representation and
    			$\Psi \in \mathcal{H}$ is a unit vector cyclic for $\pi$;
    			\item[b)] $\omega(a) = \langle \Psi | \pi(a) \Psi \rangle$;
    		\end{enumerate}
    		then there exits a unitary operator $\hat{U}:\mathcal{H}\rightarrow\Ho$ such that $\Psi = \hat{U}\Psi_\omega$ and 
    		$\pi_\omega(a) = \hat{U}\pi(a)\hat{U}^{-1}$ for any $a \in \A$.
    		
    	\end{enumerate}
    \end{theorem}
    Note that in general, $\Hi_\omega \neq \Hi_{\omega'}$ for $\omega \neq \omega'$. If $\Hi_\omega$ is finite dimensional, then $\A$ is also a von Neumann algebra; in the infinite dimensional case, this is not true anymore. Because of this theorem, we will always use algebras of operators over some Hilbert space instead of abstract objects.
    For completeness, we mention that a GNS theorem for $^*$-algebras with unit is also available (see Th. 14.20 in \cite{moretti2013spectral} ). The contents of such a theorem are more or less the same of the GNS theorem presented here. However, it allows to represent elements of a $^*$-algebra with unbounded operators which are closable over a state dependent domain $\mathcal{D}_{\omega}$. This version of the GNS theorem allows to threat in a more rigorous way unbounded random variables using unbounded operators over some Hilbert space, as mentioned at the end of section \ref{sec1e}. If $\A$ is commutative, we have another result which allows to represent abstract $C^*$-algebras with continuous functions over some space: the commutative Gel'fand-Naimark theorem. 
    \begin{theorem}
    	Any commutative $C^*$-algebra with unit $\A$ is $^*$-isomorphic (i.e. the involution is preserved under the isomorphism) to the commutative $C^*$-algebra with unit of continuous functions on $\Delta(\A)$, $C(\Delta(\A))$ (which is $C^*$ with respect to the norm $\|\cdot\|_\infty$), where
    	\begin{equation*}
    		\begin{split}
    			\Delta(\A):= \{\phi: \A\rightarrow\Comp \mbox{ }|&\mbox{ } \phi(ab) = \phi(a)\phi(b) \mbox{ }\forall a,b \in A, \\
    			& \mbox{ } \phi \mbox{ non trivial}\}.
    		\end{split}
    	\end{equation*} Such $^*$-isomorphism (called Gelfand's transform) is isometric.
    \end{theorem}
    When $\A$ is an abelian von Neumann algebra (hence also $C^*$), a similar result holds by using an algebra of measurable functions over some space (this is exactly the content of theorem \ref{th111}). Note that the GNS theorem holds also for the commutative case but only in the abelian case we can construct the measure-theoretic probability space. This fact has important consequences on the concrete interpretation of algebraic probability spaces.
    
    \subsection{Some effect of non-commutativity}
    
    Let us discuss some differences between the commutative and the non-commutative case, which are relevant for quantum theory, but the list of differences does not end here. 
    \begin{enumerate}
    	\item[i)] \emph{The lattice of projectors.} Given a $^*$-algebra $\A$, we call $p \in \A$ orthogonal projector if $p = p^* = p^2$ and the set of projectors on $\A$ will be labeled by $\mathcal{P}(\A)$. From the abelian case, we have seen that the $\sigma$-algebra of the associated probability space can be constructed from this structure ($\mathcal{P}(\A) = \tilde{\E}$, in section \ref{sec1b}). From the mathematical logic point of view, this means that in the abelian case $\mathcal{P}(\A)$ has the structure of a distributive lattice (i.e. a Boolean lattice which is always isomorphic to a Boolean $\sigma$-algebra). In the non-commutative case, this structure changes: $\mathcal{P}(\A)$ has, in general, the structure of an orthomodular lattice (modularity depends on the type of factor of $\A$). The practical consequence is that we cannot interpret the propositions about \textquotedblleft non-commutative random phenomena" using ordinary propositional calculus (the logical connectivities AND and  OR are problematic) which is exactly what happens in quantum logic. 
    	\item[ii)] \emph{The CHSH inequality.} Consider two von Neumann algebras $\A$ and $\B$ (which are automatically $C^*$) that are mutually commuting and $\A,\B \subset \boundo$. Let $\omega$ be a normal state for both algebras (hence a positive normalized linear functional from $\boundo$ to $\Comp$) and define
    	\begin{equation*}
    		\beta(\omega,\A,\B):= \sup\omega \left( a_1[b_1+b_2] + a_2[b_1 - b_2] \right)
    	\end{equation*}
    	where the $\sup$ is taken over all $a_1,a_2 \in \A$ and $b_1,b_2 \in \B$ having norm less than 1. Then if at least one of these two algebras is abelian one can prove that $\beta(\omega,\A,\B) \leqslant 2$ for all states $\omega$. When both $\A$ and $\B$ are non-abelian, then this bound can be violated: it is known that the maximal violation is $\beta(\omega,\A,\B) = 2\sqrt{2}$ \cite{cirel1980quantum}. The degree of violation depends on the type of algebra: for two mutually commuting, non-abelian von Neumann algebra, if the Schlieder property holds \cite{redei2007quantum} (i.e. $ab=0$ for $a \in \A$ and $b \in \B$ implies either $a=0$ or $b=0$) then there exists a normal state which maximally violate the inequality. This is nothing but the well known CHSH inequality of quantum mechanics \cite{clauser1969proposed}.
    	\item[iii)] \emph{Dispersion free state.} Let $\A$ be a von Neumann algebra, we say that a state $\omega$ is \emph{dispersion-free} if $\omega((a-\omega(a))^2) = 0$ for all $a \in \A$. In the abelian case, a pure state can be characterised as the states for which $\omega(ab) = \omega(a)\omega(b)$ holds for any $a,b \in \A$. Cleary, the pure states in abelian case are dispersion-free. In the non-abelian case dispersion-free states do not exist. In quantum mechanics this is a well known fact, and it is called Heisenberg uncertainty principle. We will use this fact in the next section to study a possible characterisation of non-commutativity.
    \end{enumerate}
    Other differences which are relevant from the physical point of view, between commutative and non-commutative case are, for example, the way one composes two algebras, or the algebraic generalisation of the notions of independence and conditional expectation (see \cite{redei2007quantum} and references therein for a detailed discussion). 
	
	\section{Entropic uncertainty relations}\label{section4}
	
	The non-existence of dispersion-free states in a non-commutative probability space suggests that we cannot have delta-like marginals (of some joint probability distribution) for all the random variables of our algebra. Following this intuitive idea, we introduce a natural measure of the \textquotedblleft spread" of a given probability distribution and then we discuss how this measure behaves in presence of non-commuting random variables.
	
	\subsection{Entropy in information theory}
	
	A natural measure we can use to quantify the spread of a given probability distribution is the Shannon entropy. Such entropy is the basic notion of classical information theory and for this reason it is sometimes claimed (especially in quantum physics \cite{brukner2001conceptual}) that it cannot be used for non-commutative probability spaces. From the mathematical point of view, this claim is not true, simply because any non-commutative algebra always admits a commutative sub-algebra where classical information theory can be applied. In addition, the Shannon entropy is not sensitive to the origin of probability \cite{timpson2003supposed}: it is associated to a \emph{single} random variable. 
	
	We will introduce the Shannon entropy as done in \cite{nielsen2000quantum}, which is different to Shannon's original approach. Na\"{i}vely speaking, information quantifies a number of things we do not know about a given random phenomenon. In other words, information quantifies the unexpectedness of an event $E$ relative to a random variable $X$. Let $I_X(E)$ be a measure of this unexpectedness; it is reasonable to require that
	\begin{enumerate}
		\item[i)] $I_X(E)$ is a function of the probability of $E$ to occur, and not directly a function of the event $E$;
		\item[ii)] $I_X(E)$ is a smooth function of the probability;
		\item[iii)] if $E$ and  $F$ are two disjoint events (hence independent), then $I_X(E,F) = I_X(E) + I_X(F)$.
	\end{enumerate}
	It is not difficult to see that $I_X(E) = k\log_b(P(E))$ fulfils the three requirements. Typically  $k =-1$ and $b = e$ are chosen, and this function is called \emph{information content} of the event $E$. The Shannon entropy can be thought as the expectation value of the information content of the elementary events ($E = \{\omega\}$), i.e. $H(X) := \Ex[I_X]$. Consider a discrete random variable taking values over a discrete set $\{x_1,\cdots, x_N\}$, then $H(X)$ is just
	\begin{equation}\label{SE}
		H(X) = - \sum_{i=1}^N p_i \log p_i
	\end{equation}
	where $p_i := P[X= x_i]$ where $x_i$ is one of the possible outcomes of the random variable $X$. Note that $H(X)$ remains well defined even for $N = \infty$, as one can prove by induction. To better understand how $H(X)$ quantifies the spread of a distribution, let us consider the case of a certain event (determinism). Suppose we know that the event $E:= \{X = k\}$ is always true. Then clearly $p_i = \delta_{ik}$, which gives $H(X) = 0$: the event is certain so our unexpectedness is zero (note we assumed $0\log0 =0$, as typically done in information theory). Since $-x\log x$ is always positive for $x \in [0,1]$ it is not difficult to understand that $H(X) =0$ only for delta-like distributions. In addition $iii)$ suggests that the more elementary events contribute to $H(X)$, namely the more elementary events have non zero probability, the lager its value will be. In this sense we can use $H(X)$ to quantify the spread of a probability distribution.
	
	As already observed at the beginning of this section, the only requirement needed on $\{p_i\}$ in order to define $H(X)$ is that they come from a $\sigma$-additive, normalised measure, which happens in any algebraic probability space ($\sigma$-additive means that the measure remains finite even for countable unions of events). In the non-commutative case some usual properties of $H$ do not hold: as a rule of thumb, all properties which depend on vectors of random variables (like $(X,Y)$) should be checked with care.
	
	\subsection{What is an entropic uncertainty relation?}\label{mimik}
	
	Entropic uncertainty relations are a way to introduce an uncertainty principle for generic observables in quantum mechanics. Here we will review the known bounds which are interesting for our discussion. The results presented here can be found in \cite{bialynicki1975uncertainty}, \cite{bialynicki2006formulation}, \cite{deutsch1983uncertainty}, \cite{partovi1983entropic}, \cite{maassen1988generalized} and \cite{krishna2002entropic}.
	
	Entropic uncertainty relations are relevant relations between self-adjoint operators in an Hilbert space. Let us start with a \textquotedblleft preliminary definition".
	\begin{definition}
		Consider a Hilbert space $\Hi$ and two self-adjoint operators on it, $\hat{A}$ and $\hat{B}$. Then if
		\begin{equation*}
			H_{\hat{\rho}}(\hat{A}) + H_{\hat{\rho}}(\hat{B}) \geqslant C \mspace{50mu} \forall \hat{\rho} \in \tco,
		\end{equation*}
		where $C$ is a fixed positive number independent on $\hat{\rho}$, we say that $\hat{A}$ and $\hat{B}$ fulfil an \emph{entropic uncertainty relation}.
	\end{definition}
	In the definition above, $H_{\rho}(\hat{A})$ is the Shannon entropy computed with the probability distribution $\mu_{\hat{\rho}}(\cdot) = \Tr{ \hat{\rho}\hat{P}^{(\hat{A})}(\cdot)}$, where $\hat{P}^{(\hat{A})}(\cdot)$ is the PVM associated to $\hat{A}$. The same holds for $\hat{B}$. We can clearly see why this definition should be taken with care: according to definition \ref{SE}, what is the Shannon entropy if the spectrum of the operator is continuous? We will provide a more rigorous definition in the next section, for the moment we just observe that if the Hilbert space is finite dimensional this definition works (because all operators are compact). Typically, in quantum information, one is  interested in finding the \emph{bound} $C$ (i.e. in the minimisation problem $\min_{\hat{\rho} \in \tco} [ H_{\hat{\rho}}(\hat{A}) + H_{\hat{\rho}}(\hat{B}) ]$).
	
	As a first example of the aforementioned bound (i.e. of entropic uncertainty relation), let us consider the following theorem \cite{maassen1988generalized}.
	\begin{theorem}\label{ThEUR1}
		Let $H_{\hat{\rho}}(\hat{A})$ and $H_{\hat{\rho}}(\hat{B})$ be the Shannon entropies associated to two non-degenerate self-adjoint operators $\hat{A}$ and $\hat{B}$ over a finite dimensional Hilbert space $\Hi$. Assume that $\{|\phi_a\rangle\}_{a \in \sigma(\hat{A})}$ and $\{|\psi_b\rangle\}_{b \in \sigma(\hat{B})}$ are basis of eigenvectors of $\hat{A}$ and $\hat{B}$ respectively. Then $\forall \hat{\rho} \in \tco$
		\begin{equation*}
			H_{\hat{\rho}}(\hat{A}) + H_{\hat{\rho}}(\hat{B}) \geqslant -2\log (\max_{a,b} |\langle \phi_a| \psi_b \rangle|).
		\end{equation*}
	\end{theorem}
	We can see that, if the scalar product between eigenvectors is less than 1 (i.e. $\hat{A}$ and $\hat{B}$ cannot be diagonalised at the same time, $[\hat{A},\hat{B}] \neq 0$) the bound is non-zero. This result can be generalized to the case of POVMs (which are defined as PVMs except that $ii)$ in definition \ref{PVMdef} is not required to hold), which encodes, as special case, that of degenerate operators \cite{krishna2002entropic}.
	Because we do not need all this generality, we consider the PVM case only, which can be obtained from theorem \ref{ThEUR1} by replacing the argument of the logarithm with $\max_{ab}( \| \hat{P}^{(\hat{A})}_{a}\hat{P}^{(\hat{B})}_{b}\|)$ (here $\hat{P}^{(\hat{A})}_{a}$ is the projector on the eigenspace associate to $a \in \sigma(\hat{A})$; same for $\hat{P}^{(\hat{B})}_{b}$). It is worth to say that theorem \ref{ThEUR1}, and its generalisations, is a consequence of the Riesz-Thorin interpolation theorem \cite{folland2013real} for $L_p$-spaces. From the physical point of view  this means that no physical assumption is needed to derive this theorem: in this sense it does not depend on the physical interpretation of the mathematical objects. Also in infinte dimensional Hilbert spaces we have a similar theorem. Nevertheless this time we need to face the problem that operators do not admit in general only a point spectrum. To include also the continuos-spectrum case, avoiding to introduce the \textquotedblleft continuous version" of the Shannon entropy (i.e. the differential entropy \cite{ihara1993information}, which is not properly a generalisation) the idea is simply to partition the spectrum. Given a generic operator $\hat{A}$ on an infinite dimensional Hilbert space $\Hi$, a \emph{partition of the spectrum} is a collection of set $\{E_i\}_{i \in I}$ such that $\sigma(\hat{A}) = \cup_{i \in I} E_i$. Given this partition of the spectrum and $\hat{\rho} \in \tco$, we can associate to it a set of probabilities $\{p^{(\hat{\rho})}_i\}_{i \in I}$ computed via the formula $p^{(\hat{\rho})}_i = \Tr{\hat{\rho}\hat{P}^{(\hat{A})}(E_i) }$. Using this distribution we can compute $H_{\hat{\rho}}(\hat{A})$. Then we have the following theorem \cite{partovi1983entropic}.
	\begin{theorem}\label{ThEUR2}
		Let $H_{\hat{\rho}}(\hat{A})$ and $H_{\hat{\rho}}(\hat{B})$ be the Shannon entropies associated to self-adjoint operators $\hat{A}$ and $\hat{B}$ over (possibly infinite dimensional) Hilbert space $\Hi$. Assume that $\{E_i\}_{i \in I}$ and $\{F_j\}_{j \in J}$ are two different partitions of the spectrum of $\hat{A}$ and $\hat{B}$, respectively. Then $\forall \hat{\rho} \in \tco$
		\begin{equation*}
			H_{\hat{\rho}}(\hat{A}) + H_{\hat{\rho}}(\hat{B}) \geqslant 2\log \left(\frac{2}{\sup_{i,j} \| \hat{P}^{(\hat{A})}(E_i) + \hat{P}^{(\hat{B})}(F_j) \|}\right).
		\end{equation*}
	\end{theorem}
	Again, we can see that if $\hat{A}$ and $\hat{B}$ commute, the RHS vanishes (since $1 \leqslant \sup_{i,j} \| \hat{P}^{(\hat{A})}(E_i) + \hat{P}^{(\hat{B})}(F_j) \| \leqslant 2$ and the upper bound is reached if and only if $\hat{P}^{(\hat{A})}(\cdot)$ and $\hat{P}^{(\hat{B})}(\cdot) $ have common eigenvectors).
	
    Using the two theorems presented here above, we are able to relate non-commutativity between operators and the probability measures associated with them (i.e. the spectral measure) using entropic uncertainty relations. In the next section, we will formalize these facts in a $C^*$-probability space, proving that there is a link between the non-commutativity of the algebra and the properties of the probability measures associated with states on it which can be characterized using entropic uncertainty relations.
	
	\subsection{Algebraic generalisation}
	
	In this section we will extend the definition of the Shannon entropy to a generic $C^*$-algebra. Consider a $C^*$- probability space $(\A,\omega)$. Using the GNS theorem we may equivalently consider the triple $(\Hi_\omega, \pi_\omega, \Psi_\omega)$. For any self-adjoint element  $a \in \A$, we may consider the bounded operator $\hat{A}_\omega:=\hat{\pi}_\omega(a)$ acting on $\boundoo$. The spectral theorem ensures that there exist a PVM $\{\hat{P}^{(\hat{A}_\omega)}(E)\}_{E \subset \sigma(\hat{A}_\omega)}$ associated to $\hat{A}_{\omega}$, thus the probability that $a$ takes value in $E$ is $\langle \Psi_\omega| \hat{P}^{(\hat{A}_\omega)}(E) \Psi_\omega \rangle$. Nevertheless we cannot use this probability directly in the definition of the Shannon entropy because in general the spectrum may have a continuous part. It is a known fact that if $\hat{A}$ is a bounded self-adjoint operator its spectrum can be split as $\sigma(\hat{A}) = \sigma_p(\hat{A}) \cup \sigma_c(\hat{A})$, where $\sigma_p(\hat{A})$ and $\sigma_c(\hat{A})$ are respectively the point and the continuous part of the spectrum. Note that at the algebraic level the classification of the spectrum may depend on the state $\omega$. To introduce a well defined notion of entropy at the algebraic level, we have to find a way to deal with the continuous part of the spectrum. Mimicking what we did in section \ref{mimik}, we introduce a  partition of the continuous part of the spectrum $\{E_i\}_{i \in I_\omega}$ (we always assume  $I_\omega$ at most countable). Let us label with $\varepsilon$ a generic partition, then given $\varepsilon$ we can always construct a probability distribution $\{p_i^{(\omega,\varepsilon)}\}_{i \in \sigma_p(\hat{A}_\omega) \cup I_\omega}$ for $a \in \A$ as
	\begin{equation*}
		p_i^{(\omega,\varepsilon)} := 
		\begin{cases}
			\langle \Psi_\omega| \hat{P}^{(\hat{A}_\omega)}(\{i\}) \Psi_\omega \rangle &\mspace{30mu}\mbox{ if } i \in \sigma_p(\hat{A}_\omega) \\
			\langle \Psi_\omega| \hat{P}^{(\hat{A}_\omega)}(E_i) \Psi_\omega \rangle &\mspace{30mu}\mbox{ if } i \in I_\omega.
		\end{cases}
	\end{equation*}
	Note that these probabilities clearly depend on the partition chosen, as well as on the state. Using the probability distribution constructed in this way, we can apply without problems the definition of the Shannon entropy to any self-adjoint element of $\A$.
	\begin{definition}
		Let $(\A,\omega)$ be a $C^*$-probability space. Fix a partition $\varepsilon$ and constructs for some self-adjoint $a \in \A$ the probability distribution $\{p_i^{(\omega,\varepsilon)}\}_{i \in \sigma_p(\hat{A}_\omega) \cup I}$, where $\hat{A}_\omega = \hat{\pi}_\omega (a)$. Then the \emph{$\varepsilon$-Shannon entropy} of $a \in \A$ is given by
		\begin{equation*}
			H_\omega(a;\varepsilon) := - \sum_{i \in \sigma_p(\hat{A}_\omega) \cup I_\omega} p_i^{(\omega,\varepsilon)} \log p_i^{(\omega,\varepsilon)}
		\end{equation*}
	\end{definition}
	Since the probabilities depend on the partition, the entropy depends also on the partition of the spectrum as well. Thanks to this definition we can define in a proper manner an entropic uncertainty relation in an algebraic contest. 
	\begin{definition}
		Let $\A$ be a $C^*$-algebra and consider two random variables $a,b \in \A$ on it. Choose two partitions (different in general) $\varepsilon$ and $\delta$ for $a$ and $b$, respectively. If for any $\omega$,
		\begin{equation*}
			H_\omega(a;\varepsilon) + H_\omega(b;\delta) \geqslant C(\varepsilon,\delta),
		\end{equation*}
		where $C(\varepsilon,\delta) \in \Rea^+/\{0\}$ is a constant which may depend on the partitions but not on the state, we say that $a$ and $b$ fulfil an \emph{$(\varepsilon,\delta)$-entropic uncertainty relation}.
	\end{definition}
	Note that the important part of this definition is the independence on the state of the constant $C(\varepsilon,\delta)$: the LHS is bigger than this constant \emph{for any possible state}. We need to introduce an ordering relation between partitions of the spectrum, which is nothing but the notion of  \textquotedblleft finer partition". 
	\begin{definition}
		Let $\varepsilon = \{E_i\}_{i\in I}$ and $\varepsilon'=\{E'_j\}_{j \in J}$ be two partitions. We say that $\varepsilon$ \emph{is finer than} $\varepsilon'$, written $\varepsilon \subset \varepsilon'$, if
		\begin{enumerate}
			\item[i)] $E'_j = \cup_{i \in I_j} E_i$ for some $I_j \subset I$;
			\item[ii)] $I = \cup_{j \in J} I_j$.
		\end{enumerate}
	\end{definition}
	Intuitively, a partition $\epsilon$ is finer than a partition $\epsilon'$, if combining in a suitable way the sets of $\epsilon$, we can construct all the sets of $\epsilon'$. The requirements $i)$ and $ii)$ are simply the conditions under which this combination is possible. In what follows, if we need to talk repeatedly of two partitions, say $\varepsilon$ and $\delta$, we will use the symbol $(\varepsilon,\delta)$. The writing $(\varepsilon,\delta) \subset (\varepsilon',\delta')$ means $\varepsilon \subset \varepsilon'$ and $\delta \subset \delta'$. At this point we may state the following theorem,  which relates the non-commutativity of the $C^*$-algebra and the presence of entropic uncertainty relations. 
	\begin{theorem}\label{mytheo}
		Let $\A$ be a $C^*$-algebra and take two self-adjoint elements $a,b \in \A$. If for two partitions $(\varepsilon, \delta)$ an entropic uncertainty relation holds, namely
		\begin{equation*}
			H_\omega(a;\varepsilon) + H_\omega(b;\delta) \geqslant C(\varepsilon,\delta),
		\end{equation*}
		with $C(\varepsilon,\delta)>0$, and this happens for any possible state $\omega$ over $\A$, then $[a,b] \neq 0$.
	\end{theorem}
	\begin{proof}[Idea of Proof]
		We explain the idea of the proof, whose mathematical details can be found in Appendix B.
		
		We have already seen that dispersion-free states do not exist in a non-commutative probability space. Consider a non-abelian $C^*$-algebra $\A$, and $a,b \in \A$, such that $[a,b]\neq0$. Given a state $\omega$, let $\hat{A}_\omega = \hat{\pi}_\omega(a)$ and $\hat{B}_\omega = \hat{\pi}_\omega(b)$ be the two associated GNS representations acting on $\Ho$. Assume that the spectrum is purely continuous for such representations, the discrete case can be thought as a sub-case of this. Since $[a,b]\neq0$, we cannot find a state $\omega$, which has a delta-like probability distributions (i.e. spectral measures) for both $\hat{A}_\omega$ and $\hat{B}_\omega$, for any possible partitions of the two spectra we can consider. The best we can do, is to choose $\omega$ which has a delta-like probability distribution for \emph{only one} of the two random variables: hence, we are in the situation of Figure \ref{fig:0}.
		\begin{figure}[h!]
			\includegraphics[scale=0.5]{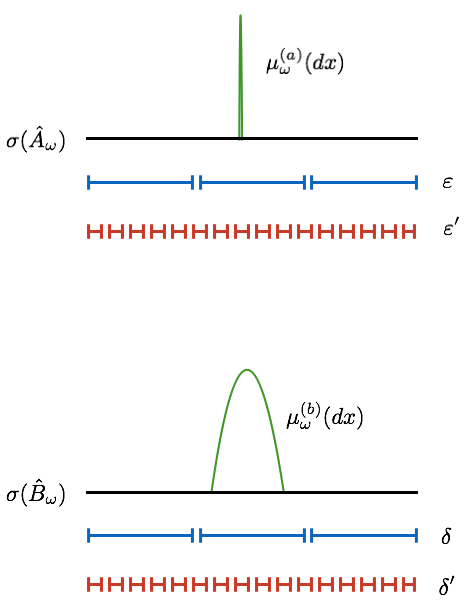}
			\caption[sg,small]{The case of a state which is not dispersion-free. As one can see we can have for all possible partitions, a delta-like probability distribution for the algebraic random variable $a$. This cannot happen for $b$, when $[a,b]\neq 0$: there are partitions for which the probability distribution of $b$ cannot be delta-like. }
			\label{fig:0}
		\end{figure}
		Now, take $(\varepsilon,\delta)$ and $(\varepsilon',\delta')$ such that $(\varepsilon',\delta') \subset (\varepsilon,\delta)$. If $\omega$ induces the two probability distributions in the picture, we can see that:
		\begin{enumerate}
			\item[i)] There are partitions where the two probability distributions $\{p_i^{(\omega,\varepsilon)}(a)\}_{i \in I}$ and $\{p_i^{(\omega,\delta)}(b)\}_{i \in I}$, have a delta-like shape (i.e. all the probabilities are $0$ except for one set of the partition). This is the case of the partitions $(\varepsilon,\delta)$ in the Figure \ref{fig:0}.
			For these partitions, we have no entropic uncertainty relations.
			\item[ii)] There are partitions where only one of the two probability distributions still have a delta-like shape. This is the case of the partitions $(\varepsilon',\delta')$ in the Figure \ref{fig:0}. In this case, we have an entropic uncertainty relation.
		\end{enumerate}
		Hence, if dispersion-free states do not exist (i.e. the algebra is non-commutative) we can find a partition for which an entropic uncertainty relation holds. It is not difficult to understand that if an entropic uncertainty relation holds for a couple of partitions, it also holds for any finer couple of partitions. On the other hand, it is not difficult to see that if an entropic uncertainty relation is found for a partition $(\varepsilon,\delta)$, automatically it holds for all the finer partitions $(\varepsilon',\delta') \subset (\varepsilon,\delta)$. Thus there are no dispersion-free states on the algebra, which means it is not abelian.
	\end{proof}
	Note that this theorem gives a way to test if two algebraic random variables commute or not, using purely probabilistic concepts. This result generalizes in the algebraic contest the content of the theorems \ref{ThEUR1} and \ref{ThEUR2} seen in the previous section. Note that the constant $C(\varepsilon,\delta)$ may depend on the partition, while, in the theorems \ref{ThEUR1} and \ref{ThEUR2}, this dependence is absent. This fact is suggesting that the dependence on the partition is more an artifact due to the definition we gave for $\varepsilon$-Shannon entropy in a $C^*$-algebraic contest, instead of something deeper. Hence such dependence could be eliminated, but we were not able to do so. Finally note that this theorem does not say anything about the bound (contrary to theorems \ref{ThEUR1} and \ref{ThEUR2}), but it asserts only that if it exists for all states then the algebra is non-commutative.
	
	The main difficulty in the use of this test for non-commutativity, is that the LHS of the inequality must be varied over all the possible states. Fortunately, a further simplification can be done. Consider a state $\omega: \A \rightarrow \Comp$ and let $\mathcal{S}(\A)$ be the set of all states on $\A$. Then we say that the state is \emph{pure} if it cannot be written as convex combination of other states (i.e. $\nexists \omega_1,\omega_2$ such that $\omega(\cdot) = \lambda \omega_1(\cdot) + (1-\lambda)\omega_2(\cdot)$ for some $\lambda \in [0,1]$), otherwise it is said to be \emph{mixed}. Let $\mathcal{S}_p(\A)$ denote the set of all the pure states on $\A$.
	\begin{corollary}\label{mycor}
		Consider a $C^*$-algebra $\A$ and take two self-adjoint elements $a,b \in \A$. If an entropic uncertainty relation between $a$ and $b$ holds for all $\omega \in \mathcal{S}_p(\A)$ , then it also holds for any state in $\mathcal{S}(\A)$.
	\end{corollary}
	\begin{proof}
		See appendix C.
	\end{proof}
	Hence it is sufficient to check this relation by varying $\omega$ on the pure states only.

    \section{Non-commutativity from ordinary measure-theoretic probability}\label{section5}
    
    We have seen that if we want to model random phenomena, we can use two (apparently) different mathematical structures: a measure space, $\PS$, or an algebra with a state $(\A, \omega)$. On the other hand, random phenomena in the subatomic world can only be described using algebraic probability, since the non-commutative behavior seems to play a fundamental role. Here we want to discuss a possible method to obtain a non-commutative behavior of the probability starting from a collection of probability spaces.
    
    \subsection{The general method in the algebraic setting}
    
    Before we explain the method in the Hilbert space setting, let us discuss the idea from the algebraic point of view. Suppose we have two real random variables $a$ and $b$. Instead of describing them in the measure-theoretic language, we wish to describe them using the abelian algebras that they generate, say $\A_a $ and $\A_b$ respectively. This means that the algebra $\A_a$ is the abelian algebra generated by the identity, $a$ and all its polynomial $p(a)$. The same for $\A_b$. On these two algebras, we can define states: we label by $\omega_a$ states on $\A_a$ and by $\omega_b$ states on $\A_b$. Now we assume the following: \emph{there exist a 1-1 map between states on $\A_a$ and states on $\A_b$}. This means that to a given state $\omega_a$ on $\A_a$, we can associate in a unique way a state $\omega_b$ on $\A_b$. This map allows neglecting the labels $a$ and $b$ in the symbol of the state $\omega$. Let us now set $\A$ as the smallest $C^*$-algebra containing both $\A_a$ and $\A_b$ as subalgebras (i.e. the algebra generated by the identity, $a$, $b$ and polynomials $p(a,b)$). By theorem \ref*{mytheo}, if we can prove that $H_{\omega}(a) + H_{\omega}(b) \geqslant D$ for all $\omega$ (the dependence on the partitions is omitted for simplicity), we know that $\A$ is a non-abelian algebra. Since $\A$ is non-abelian, the GNS theorem allows to represent it as an algebra of bounded operators on a suitable Hilbert space. We cannot represent $\A$ as an algebra of functions. Thus, starting from two ordinary random variables defined in two \emph{different} probability spaces, we end up with an Hilbert space description where both random variables are present as operators, but they do not commute.
	
	\subsection{The Hilbert space structure from the entropic uncertainty relation}\label{sec5a}
	
	Previously we presented in algebraic setting a method to obtain a non-commutative probability space starting from two ordinary measure-theoretic probability spaces. In this section we will explain how to construct a concrete algebraic probability space (i.e. already represented on an Hilbert space) starting from the probability spaces of two random variables, assuming that they fulfill an entropic uncertainty relation. To keep the discussion simple, we restrict ourself to finite discrete random variables.
	
	Let $X:(\Omega,\E,P) \rightarrow (M,\mathcal{M})$ and $Y:(\Omega',\E',P') \rightarrow (N,\mathcal{N})$ be two discrete random variables and as usual $\mu_X:= P \circ X^{-1}$ and $\nu_Y:= P' \circ Y^{-1}$ label their probability distributions. We assume the following conditions:
	\begin{enumerate}
		\item[i)] we have a 1-1 map between $P$ and $P'$, i.e. to each probability distribution $\mu_X$ for $X$ we can associate a corresponding probability distribution $\nu_Y$ for $Y$ and viceversa;
		\item[ii)] $M$ and $N$ have the same cardinality, i.e. $X$ and $Y$ have the same number of possible distinct outcomes;
		\item[iii)] $X$ and $Y$ fulfil an entropic uncertainty relation, namely for any $\mu_X$ and $\nu_Y$
		\begin{equation*}
			H(X) + H(Y) \geqslant D
		\end{equation*}
		with $D>0$.
	\end{enumerate}
	In section \ref{sec1e}, we have seen that a consistent way to represent a random variable on an Hilbert space is obtained by using the spectral representation theorem and the spectral decomposition theorem. Thus, given the random variable $X$, we can construct the operator
	\begin{equation*}
		\hat{T}_X := \sum_{x \in M} x | x \rangle \langle x|
	\end{equation*}
	defined on the Hilbert space
	\begin{equation*}
		\Hi_X := \bigoplus_{x \in M} \Hi_x.
	\end{equation*}
	By construction $\sigma(\hat{T}_X) = M$ and $\{|x\rangle\}_{x \in M}$ is a basis of $\Hi_X$. The assumption $i)$ ensures that, in general, the operator representing the random variable $X$ cannot be used to describe also the random variable $Y$. More precisely,  as we have seen in section \ref{PinH}, if $\{|x\rangle\}_{x \in M}$ is the basis on which $\hat{X}$ is diagonal, we can represent over this basis all the random variables that are functions of $X$. Hence, thanks to the assumption $i)$, we can go beyond the simple case of $X = f(Y)$ (or $Y=g(X)$), where the map between $P$ and $P'$ is given by a simple change of variables. The random variable $Y$, being defined on a different probability space, cannot be seen in general as a function of $X$. Repeating the whole construction for the random variable $Y$, also in this case we can define an operator
	\begin{equation*}
		\hat{S}_Y = \sum_{y \in N} y | y \rangle \langle y |
	\end{equation*}
	on the Hilbert space
	\begin{equation*}
		\Hi_Y := \bigoplus_{y \in N} \Hi_y.
	\end{equation*}
	Note that this Hilbert space is not in general $\Hi_X$. Again $\sigma(\hat{S}_Y) = N$ and $\{|y\rangle\}_{y \in N}$ is a basis of $\Hi_Y$ by construction. The assumption $ii)$ ensures that the two Hilbert spaces have equal dimension, and so there exists a unitary map $\hat{U}: \Hi_X \rightarrow \Hi_Y$. This means that we can map the operator $\hat{S}_Y$ on $\Hi_X$ and $\hat{T}_X$ on $\Hi_Y$. Let us consider the first case, since the second is equivalent. The operator representing $Y$ on $\Hi_X$ is
	\begin{equation*}
		\begin{split}
			\hat{T}_Y :&= \hat{U} \hat{S}_Y \hat{U}^* \\
			&= \hat{U} \sum_{y \in N} y | y \rangle \langle y | \hat{U}^* \\
			&= \sum_{y \in N} y \hat{U}| y \rangle \langle y | \hat{U}^*.
		\end{split}
	\end{equation*}
	Let us set $|Uy\rangle := \hat{U}|y \rangle$ and note that the operator $\hat{T}_Y$ is diagonal in this basis. Since unitary transformation maps a basis into a basis, also $\{|Uy \rangle\}_{y \in N}$ is a basis and in particular it is the image under $\hat{U}$ of the basis in which $\hat{S}_Y$ is diagonal. At this point, the key observation is that if the assumption $iii)$ is true, then the basis $\{|x\rangle\}_{x \in M}$ and the basis $\{|Uy\rangle\}_{y \in N}$ do not coincide. Indeed, the entropic uncertainty relation assumed, together the theorem \ref{ThEUR1}, allows us to write that
	\begin{equation*}
		\begin{split}
			-2\log (\max_{x,y} |\langle x| Uy \rangle|) \geqslant D
		\end{split}
	\end{equation*}
	(with the equality only if one can prove that the bound is optimal) so $\max_{x,y} |\langle x | Uy \rangle| \leqslant e^{- D/2} < 1$, since $D$ is never zero. Another way to say this is that $\hat{U}$ is not the identity transformation. Note that we can reach this conclusion only because we assumed the existence of an entropic uncertainty relation: if $D=0$, then we cannot exclude that $|\langle x| Uy \rangle| = 1$ for some $x,y$ (i.e. they are the same basis). 
	
	The conclusion is that, given the entropic uncertainty relations, the two operators $\hat{T}_X$ and $\hat{T}_Y$ do not commute, thus we can describe \emph{both} random variables only on a common non-commutative algebraic probability space (i.e. with operators on an Hilbert space). How on this structure is represented the map between $P$
	and $P'$, i.e. the state, will be discussed in the next section.
	
	\subsection{Conditional probabilities and representation of states}\label{pre-final section}
	
	In the previous section, we have seen that starting from two random variables defined on two different probability spaces, if an entropic uncertainty relation holds, we can construct a non-commutative algebraic probability space where both the random variables are represented by non-commuting operators. Essential for this construction is the presence of two distinct probability space, one for each random variable. Here we want to discuss how this condition can be met in a rather simple way and the consequences of this on the map between $P$ and $P'$.
	
	Given a probability space $\PS$ and a collection of events, conditioning with respect to each of these events, generates a collection of probability spaces. More precisely, conditional probability in measure-theoretic setting is defined via the Bayes formula
	\begin{equation*}
		P_C(A) := P(A|C) = \frac{P(A \cap C)}{P(C)}  \mspace{30mu} A,C \in \E
	\end{equation*}
	$P_C$ is again a probability measure on $\Omega$, but this time it depends on the event $C$ also. Given a family of events $\mathcal{C}:=\{C_i\}_{i \in I}$, then by conditioning we obtain the collection of probability spaces $(\Omega, \E(\Omega)_{C_i}, P_{C_i})_{C_i \in \mathcal{C}}$. The trivial case $\mathcal{C} = \{C\}$ coincides with the usual measure-theoretic description, however in the more general case, this collection is called \emph{contextual probability space} \cite{khrennikov2014ubiquitous} \cite{khrennikov2009contextual} \cite{khrennikov2016random} while the $C_i$s are called \emph{context}. In the general contextual probability theory, not all the context are elements of a $\sigma$-algebra (i.e. events, as in this case). This means that it is not assumed that all the contextual probability spaces are generated by conditioning. Similar notions were introduced also in \cite{holevo2011probabilistic}, where very general results are presented, and in \cite{aerts1995quantum}.
	
	Consider now two random variables $X$ and $Y$ on $\PS$ with distributions $\mu_X = P \circ X^{-1}$ and $\nu_Y = P \circ Y^{-1}$. Assume for simplicity that they are discrete. Conditioning alone is not sufficient to ensure that they are described in two different probability spaces. Indeed, since they are functions on the same probability space, after conditioning they can always be described on a probability space $(\Omega, \E(\Omega)_{C_i}, P_{C_i})$ where, from a (conditional) joint probability distribution, $\eta_{X,Y | C_i} = P_{C_i} \circ (X^{-1},Y^{-1})$, we can derive the two marginals $\mu_{X|C_i}$ and $\nu_{Y|C_i}$ describing $X$ and $Y$ (after conditioning). However, we may proceed in a different manner. Suppose that $X$ and $Y$ are two random variables on $\PS$ with fixed transition probabilities $\alpha(x,y) := P[X=x|Y=y]$ and $\tilde{\alpha}(y,x) := P[Y=y|X=x]$. The random variables $X$ and $Y$ after conditioning are described by the conditional probability distributions $\mu_{X|C_i}$ and $\nu_{Y|C_i}$. It is not difficult to see that, if we use these fixed transition probabilities, in general
	\begin{equation}\label{nobayes}
		\alpha(x,y)\nu_{Y|C_i}(y) \neq \tilde{\alpha}(y,x)\mu_{X|C_i}(x).
	\end{equation}
	In an ordinary measure-theoretic model of probability the product of the transition probability, times the marginal gives the joint probability distribution, which is symmetric under the exchange of its arguments (it is a consequence of the fact that events are subsets of the \emph{same} sample space). In our case fixing the transition probabilities and using the conditional probabilities for the two random variables, makes impossible to define a joint probability distribution. More precisely, it does not exist a joint probability distribution which has $\mu_{X|C_i}$ and $\nu_{Y|C_i}$ as marginals, and such that $\alpha(y,x)$ and $\tilde{\alpha}(y,x)$ are the two transition probabilities which can be derived from it. Hence if we fix the transition probabilities in advance, the random variables $X$ and $Y$ after conditioning must be considered to be defined on two different probability spaces in general. Another way to see this is via \emph{Bayes theorem}. From \eqref{nobayes}, one can conclude that
	\begin{equation*}
		\delta(x|Y,C_i) = \mu_{X|C_i}(x) - \sum_y \alpha(x,y)\mu_{Y|C_i}(y) \neq 0
	\end{equation*}
	which means that the Bayes theorem does not hold. This has big consequences on the representation with a single mathematical object of the two probability distributions $\mu_{X|C_i}$ and $\mu_{Y|C_i}$. Since we are not working on a single probability space, the procedure explained in section \ref{sec1e} no longer work. In fact, if we follow this procedure we can associate to $\mu_{Y|C_i}(y)$ the trace class operator $\hat{\rho}_Y = \sum_y \mu_{Y|C_i}(y) |y \rangle \langle y |$, from which we have to conclude that
	\begin{equation*}
		\begin{split}
			\mu_{X|C_i}(x) &= \Tr{\hat{\rho}_Y |x\rangle\langle x |} \\
			&= \sum_{y} |\langle x | y \rangle|^2\mu_{Y|C_i}(y).
		\end{split}
	\end{equation*}
	Interpreting $\alpha(x,y) = |\langle x | y \rangle|^2$, we can see that only if $\delta(x|Y,C_i) = 0$ the map between $\mu_{X|C_i}$ and $\mu_{Y|C_i}$ can be described in this way. When $\delta(x|Y,C_i) \neq 0$ we have to proceed in a different way. As explained in \cite{khrennikov2009contextual}, the term $\delta(x|Y,C_i)$ play the role of the interference. In ordinary quantum theory, the interference is a consequence of the presence of a non-trivial phase factor when we project a vector on a different basis.
	Under suitable conditions on the probability distributions (among with $\alpha(x,y) = \tilde{\alpha}(y,x)$) an algorithm, for the construction of the vector $|\psi\rangle \in \Hi$ and the representation of the two random variables by means of operators on $\Hi$, is available \cite{khrennikov2014ubiquitous}\cite{khrennikov2009contextual}\cite{khrennikov2016random}. It is called \emph{Quantum-Like Representation Algorithm} (or QLRA for short). However, this algorithm has some limitations. In particular, it is fully developed only for the case of random variables having two or three possible outcomes \cite{nyman2011quantum}: only in this case, the algorithm is capable to give us a non-commutative probability space. The general case it is not fully developed, despite the difficulties seems to be more in computational side rather than mathematical one. On the other hand, the method proposed here, based on entropic uncertainty relations does not have limitations regarding the kind of random variables used. It does not tell us how to find $|\psi\rangle $ explicitly but, once that conditions $i) - iii)$ of section \ref{sec5a} are fulfilled, we know that the random phenomena must be described using a non-commutative probability space. Although this is not a tremendous improvement with respect to QRLA, this method allows to study interesting situations, as we will do in \cite{LC2} and \cite{LC3}.
	
	Before to conclude this section, we want to observe the following fact. Given $(\Omega, \E(\Omega)_{C_i}, P_{C_i})_{C_i \in \mathcal{C}}$, we cannot reconstruct the original probability space $\PS$. Additional information is required: we need $P(C_i)$. In this sense, if $\mathcal{C}$ is the set of all elementary events for a random variable $Z$, i.e. all events like $C_i:=\{Z = z_i\}$, such random variable $Z$ cannot be described with the contextual probability space obtained after conditioning. In this sense, $Z$ is no longer present in the (probabilisitic) model. Because of this fact, we will also say that the random variable $Z$ was \emph{removed} from the model. Such collection of probability spaces thus represents a tool to describe a random phenomenon, after a random variable (representing some feature of such a phenomenon) is eliminated from the description. Such an elimination procedure may not always give rise to a non-commutative representation of the probability theory describing a given phenomenon. Indeed, it is not clear if such elimination procedure implies always an entropic uncertainty relation.

	\section{Conclusion}
	
	In this article, we analyzed two possible ways to mathematically describe random phenomena. The analysis suggests that one should consider the measure-theoretic formulation and algebraic formulation of probability theory as two representations of the theory of probability: the choice of one representation over the other depends on the limitations we may have on the description of the phenomenon. This is very interesting if we apply this idea to quantum theory. There are many attempts to derive quantum theory from an underlying description in which the interpretation is well defined: they go under the name of hidden variable theories. Typically in these theories, the attempt to re-obtain the statistical prediction of quantum mechanics is done by averaging over the random variables representing hidden quantities which are not under experimental control. Quantum mechanics and its mathematical formalism are seen as a  \textquotedblleft thermodynamical limit" of the underlying model.  Plenty of no-go theorems were found and they make rather difficult to derive some physically plausible underlying theory in this way. Here a different strategy is proposed, which however does not allow to identify a unique \textquotedblleft underlying" model for quantum theory. Note however that the term \textquotedblleft underlying" is, in some sense, misleading in this case: quantum theory does not arise as a thermodynamic-like theory of some deeper reality but is considered as the theory of probability of reality. The method proposed here tells us, given a model of reality, how to test if quantum mechanics is the theory of probability of the model. This is what is done in \cite{LC2} and \cite{LC3} with the goal to re-obtain non-relativistic quantum mechanics. Maybe, this change of point of view can shed some light on the basic question \emph{\textquotedblleft Why do we have to use quantum mechanics to describe microscopic phenomena?"}.
	
	\section{Acknowledgements}
	
	I would like to thanks S. Bacchi and S. Marcantoni for endless discussions, comments, and suggestions during the writing of this article. A special thanks go to G. Gasbarri and  M. Toro\v{s} for their careful reading and their useful comments. Thank you also to Prof. V. Moretti, for his patience during my MSc thesis where a first proof of the theorem \ref{mytheo} was found, and to Prof. A. Bassi, for the useful corrections, discussions and the freedom given in this first part of my Ph.D. Least but not last, big thanks goes to G. Chiruzzi, V. Confalonieri, and C. Ferrario, for their patience and support during the genesis of the ideas reported here.

	\section{Appendix}
	
	\subsection*{A - Proof of the theorem \ref{theo2}}
	
	\begin{proof} (see \cite{massen1998quantum})
		We can easily see that $\A = \{\hat{M}_f | f \in L_\infty\PS \}$ is an algebra of operators (with involution). The state $\phi$ is clearly faithful, hence what we need to prove is that $\A$ is strongly closed. Then if this is true, by the von Neumann's double commutant theorem \cite{moretti2013spectral}, $\phi$ is always a normal state. In order to prove the strong closure, let us consider a sequence of functions $\{f_n\}_{n \in \Nat} \in L_\infty \PS$ such that
		\begin{equation*}
			s-\lim_{n\rightarrow\infty} \hat{M}_{f_n} = \hat{X}
		\end{equation*} 
		where $\hat{X}$ is some operator. The above expression is equivalent to 
		\begin{equation*}
			L_2 - \lim_{n \rightarrow \infty} \hat{M}_f \psi = \hat{X}\psi \mspace{50mu} \forall \psi \in L_2\PS
		\end{equation*}
		We may always assume, without loss of generality that $\|\hat{X}\| = 1$. Now, we need to prove that $\hat{X} = \hat{M}_f$. Let us set $f(\omega) := \hat{X} 1(\omega)$, since the identity function $1(\omega) \in L_2\PS$. Now, consider the set
		\begin{equation*}
			E_\epsilon := \{\omega \in \Omega | |f(\omega)|^2 \geqslant 1 + \epsilon\} 
		\end{equation*}
		for any $\epsilon > 0$. Clearly, $E_\epsilon \in \E$ and so, using the Cauchy – Schwarz inequality and recalling that $\|\hat{X}\| = 1$, we can write
		\begin{equation*}
			\begin{split}
				P(E_\epsilon) &= \int_\Omega \chi_{E_\epsilon}(\omega) P(d\omega) = \|\chi_{E_\epsilon} (\omega)\|^2_{L_2} \\
				&\geqslant \| \hat{X}\chi_{E_\epsilon} (\omega)\|^2_{L_2} = \| f\chi_{E_\epsilon} (\omega)\|^2_{L_2} \\
				&= \int_{E_\epsilon} |f(\omega)|^2P(d\omega) \geqslant (1 + \epsilon) P(E_\epsilon)
			\end{split}
		\end{equation*}
		which implies that $P(E_\epsilon) = 0$. Since this holds for any $\epsilon > 0$, then $|f| \leqslant 1$ almost everywhere with respect to $P$, from which we conclude that $f \in L_\infty \PS$. For $g \in L_\infty \PS$, since $L_\infty \PS \subset L_2 \PS$, we can write that:
		\begin{equation*}
			\begin{split}
				\hat{X}g &= L_2-\lim_{n \rightarrow \infty}f_n g \\ 
				&=  \hat{M}_g L_2-\lim_{n \rightarrow \infty}f_n \\
				&= \hat{M}_gf = gf.
			\end{split}
		\end{equation*}
		Since $L_\infty \PS$ is dense in $L_2 \PS$ we can conclude that $\hat{X} = \hat{M}_f$. 
	\end{proof}
	
	\subsection*{B - Proof of the theorem \ref{mytheo}}
	
	Before starting the proof, let us first explain its structure and some technical facts. The proof can be divided in two parts. In the first part (\textsc{Step 1} - \textsc{Step 3}) is just a proof that dispersion free states do not exist in a non abelian algebra. Clearly it is not the first time that this fact is proved, however here we prove this fact in a probabilistic manner: we construct explicitly the random variables and the joint probability spaces. Only in the last part (\textsc{Step 4}), the entropic uncertainty relations come into play. To explicitly construct the joint probability space and the random variables, a technical point is needed.
	\begin{theorem}[Th. 9.15 \cite{moretti2013spectral}; Cp. IV, Th. 2.3 \cite{prugovecki1982quantum}] \label{JSM}
		Let $\Hi$ be a separable Hilbert space and let $\hat{A}_1, \cdots, \hat{A}_n$ be a set of self-adjoint mutually commuting bounded operators. Let $\hat{P}^{(\hat{A}_1)}, \cdots, \hat{P}^{(\hat{A}_n)}$ be the associated PVMs, then there exists a unique PVM $\hat{P}^{(\hat{\bf{A}})}$ such that
		\begin{equation*}
			\hat{P}^{(\hat{\bf{A}})}(B_1 \times \cdots \times B_n) := \hat{P}^{(\hat{A}_1)}(B_1) \cdots \hat{P}^{(\hat{A}_n)}(B_n)
		\end{equation*}
		where $B_i \in \borel$ for any $i$. $    \hat{P}^{(\hat{\bf{A}})}:\borell{n}\rightarrow\boundo$ is called \emph{joint PVM} of $\hat{A}_1, \cdots, \hat{A}_n$.  If $F:\Rea \rightarrow \Comp$ is bounded measurable function, then
		\begin{equation*}
			\int_{\Rea^n} F(x_k(\mathbf{x})) \hat{P}^{(\hat{\mathbf{A}})}(d \mathbf{x}) = \int_{\Rea} F(x_k) \hat{P}^{(\hat{A}_k)}(d x_k) = F(\hat{A}_k)
		\end{equation*}
		where $\mathbf{x} = (x_1,\cdots,x_n) \in \Rea^n$ and $x_k(\mathbf{x})$ is the $k$-th component of $\mathbf{x}$.
	\end{theorem}
	We can say that our proof is a corollary of this theorem. In particular we are interested in the consequences of it on the spectral measure. First we recall that, as a consequence of the spectral decomposition theorem, given a bounded self-adjoint operator $\hat{T}$ then its spectrum is the support of the associated PVM, i.e. $\sigma(\hat{T}) = \mbox{supp} ( \hat{P}^{(\hat{T})} )$. Now, the above theorem ensures that $\sigma(\hat{\mathbf{A}}) = \sigma(\hat{A}_1) \times \cdots \times\sigma(\hat{A}_n)$. Using the joint PVM and given a normalised $\psi \in \Hi$, we can define the \emph{joint spectral measure} simply as:
	\begin{equation*}
		\mu_\psi^{(\hat{\mathbf{A}})} (d \mathbf{x}) = \langle \psi | \hat{P}^{(\hat{A}_1)}(dx_1) \cdots \hat{P}^{(\hat{A}_n)}(dx_n) \psi \rangle
	\end{equation*}
	which is a probability measure on the probability space $(\sigma(\hat{A}_1) \times \cdots \times\sigma(\hat{A}_n), \E, \mu_\psi^{(\hat{\mathbf{A}})})$, where $\E$ is a borel $\sigma$-algebra. This is the tensor product of the probability spaces associated to the spectral measures one can construct from the single PVMs, i.e. $(\sigma(\hat{A}_k), \mathscr{B}(\sigma(\hat{A}_k)), \mu_\psi^{(\hat{A}_k)}(dx) )$. The independence properties of the probability measure $\mu_\psi^{(\hat{\mathbf{A}})}$ depend on $\psi$.
	
	Hence the existence of a joint PVM of the multiplicative form as described in the theorem, ensures the existence of a common probability space for all commuting operators when thought as random variables. If the operators $\hat{A}_1, \cdots, \hat{A}_n$ do not commute this is not anymore possible: we can always multiply them obtaining again a self-adjoint operator, the spectral theorem ensures the existence of a PVM for such product and so a probability space (i.e. a spectral measure) can be defined, but this probability space cannot be related (at least in a trivial manner, i.e. the one seen above) to the probability spaces associated to each operator.
	
	The second technical fact is the following. In general, not all the GNS representations of a $C^*$-algebra $\A$ are faithful. Faithfulness is an important property because it allows to think the whole algebra as operators over the \emph{same} Hilbert space. Luckily, there exists the (general) \emph{Gel'fand-Naimark theorem} which tells us how to construct a representation which is always faithful (the so called universal representation).
	\begin{theorem} [Th. 14.23 \cite{moretti2013spectral}]
		For any $C^*$-algebra with unit $\A$ there exists an Hilbert space $\Hi$ and an isometric $^*$-isomorphism $\Pi:\A\rightarrow \mathcal{B}_{GN}$, where $\mathcal{B}_{GN} \subset \boundo$ is a  $C^*$-sub-algebra of $\boundo$.
	\end{theorem}
	More precisely, the universal representation of $\A$ on $\Hi$ is defined as the direct sum of all the representations with respect to $\omega$, i.e. $\Pi := \oplus_\omega \pi_\omega$. The Hilbert space is defined in a similar way, i.e. $\Hi := \oplus_\omega \Ho$. This allows to always think about $\A$ as a $C^*$-sub-algebra if $\boundo$, but for technical reasons, we will need to consider the von Neumann algebra that we can construct on $\Hi$ closing $\B_{GN}$. This forces us to the following definition.
	\begin{definition}
		Given a $C^*$-algebra with unit $\A$, then $\A^{vn} := \overline{\Pi(\A)}^s$ is the closure in the strong topology of $\A$ when it is though as an algebra of bounded operators on $\Hi$, the Hilbert space of the universal representation. $\A^{vn}$ will be called \emph{strong closure of} $\A$.
	\end{definition}
	Now we can start with the proof of theorem \ref{mytheo}.
	
	\begin{proof}
		\textsc{Step 1: $[a,b] = 0 \mspace{5mu}\Rightarrow \mspace{5mu} a = f_1(c), \mspace{2mu}b = f_2(c)$ for $c \in \A^{vn}$.}\newline
		We will prove that given $a,b \in \A$, if $[a,b]=0$ then we can always find two maps $f_1$ and $f_2$ and an element of the strong closure of the algebra $c \in \A^{vn}$ such that $a = f_1(c)$ and $b = f_2(c)$. Let $\A_c[a,b]$ be the commutative sub-algebra of $\A$ generated by $a$ and $b$. 
		
		Consider a state $\omega:\A \rightarrow \Comp$, then by the GNS theorem we may represent $\A$ using bounded operators over $\Ho$. Assume that $\hat{\pi}_\omega$ is faithful (i.e. one-to-one, we will deal with the general case at the end) and let $\hat{A}_\omega = \hat{\pi}_\omega (a)$ and $\hat{B}_\omega = \hat{\pi}_\omega (b)$ be the representation of $a$ and $b$ on it. Since $[\hat{A}_\omega,\hat{B}_\omega] = 0$ then theorem \ref{JSM} ensures that there exists a joint spectral measure, say $\hat{P}^{(\hat{C})}$, which is associated to some self-adjoint bounded operator $\hat{C}$ whose existence is guaranteed by the spectral decomposition theorem. Now, take $F:\Rea \rightarrow \Comp$ as the identity function, then theorem \ref{JSM} allows us to write
		\begin{equation*}
			\begin{split}
				\int_{\sigma (\hat{A}_\omega) \times \sigma (\hat{B}_\omega)} x_1(\alpha,\beta) \hat{P}^{(\hat{C})}(d\alpha d\beta) = \\
				= \int_{\sigma(\hat{A}_\omega)} \alpha \hat{P}^{(\hat{A}_\omega)}(d\alpha) = \hat{A}_\omega
			\end{split}
		\end{equation*}
		where $x_1$ is the projector on the 1-th component of the vector $(\alpha,\beta)$. Thus we can see that, setting $f_1(\cdot) = x_1(\cdot)$, we have $\hat{A}_\omega = f_1(\hat{C})$. Clearly the same holds for $\hat{B}_\omega$. Because by construction $\hat{C}$ commutes with either $\hat{A}_\omega$ and $\hat{B}_\omega$, then $\hat{C} \in \A_c[\hat{A}_\omega,\hat{B}_\omega]$. The chosen representation $\hat{\pi}_\omega: \A \rightarrow \boundoo$ is faithful, hence we can conclude that $c \in \A_c[a,b] \subset \A \subset \A^{vn}$, $a = f_1(c)$ and $b = f_2(c)$. Assume now that the representation is not faithful. In this case we may invoke the Gel'fand-Naimark theorem and then the same arguments apply. This time $c$ may not belong to the original algebra $\A$ (it belongs to the strong closure of $\A$,  $c \in \A_c[\hat{A}_\omega,\hat{B}_\omega] \subset \A^{vn}$ in general) because $c$ this time is a generic element of $\boundo$, not of $\B_{GN}$.
		
		\textsc{Step 2: Given $\Pi: \A\rightarrow\A^{vn}$ then $\sigma_{\A^{vn}}( \Pi(a) ) \subset \sigma_\A (a)$. }\newline
		This is a technical step. We recall that given $a \in \A$, $C^*$-algebra with unit $\unit$, \emph{the spectrum of $a$} (in $\A$) from the algebraic point of view is the set $\sigma_{\A}(a) := \{ \lambda \in \Comp | \nexists (a - \lambda \unit)^{-1} \in \A\}$.
		
		Let $\tau: \A \rightarrow \A^{vn}$ be a unit-preserving $^*$-homomorphism between $\A$ and its strong closure (it cannot be an isomorphism). An example is $\tau(\cdot) = \Pi(\cdot)$, i.e. the universal representation itself. Let $\unit$ and $\unit_{vn}$ be the identities of $\A$ and $\A^{vn}$, respectively. Then by definition of $\tau$, $\unit_{vn} = \tau(\unit)$ and so given $\lambda \in \Comp$ we have
		\begin{equation*}
			\unit_{vn} = \tau( (a - \lambda\unit)^{-1}(a - \lambda \unit) ) = \tau((a - \lambda \unit)^{-1})(\tau(a) - \lambda \unit_{vn})
		\end{equation*}
		and so we have that $(\tau(a) - \lambda \unit_{vn})^{-1} = \tau((a - \lambda \unit)^{-1})$. Now, if $\lambda \in \sigma_{\A^{vn}}(a)$, then $\nexists (\tau(a) - \lambda \unit_{vn})^{-1} \in \A^{vn}$, which is possible if and only if $\nexists (a - \lambda \unit)^{-1} \in \A$ which means that $\lambda \in \sigma_{\A}(a)$. The converse is not true in general. Thus $\sigma_{\A^{vn}}(\tau(a)) \subset \sigma_{\A}(a)$.
		
		\textsc{Step 3: $[a,b] = 0 \mspace{5mu}\Rightarrow \mspace{5mu} \exists\omega $ such that $ \mu_\omega^{(a)}(\cdot) =\delta_\alpha(\cdot) $ and $\mu_\omega^{(b)}(\cdot) = \delta_\beta (\cdot)$. }\newline
		Because $[a,b]=0$, we know that there exists a probability space where $a$ and $b$ are ordinary random variable. Let $\mu^{(a)}_\omega(dx)$ and $\mu_\omega^{(b)}(dx)$ be the spectral measure of $a$ and $b$ for a given state $\omega$ on $\A$. We want to prove that there exists a state $\omega$ such that $\mu_\omega^{(a)}(dx) = \delta_\alpha(dx)$ and $\mu_\omega^{(b)}(dx) = \delta_\beta(dx)$ for some suitable $\alpha \in \sigma_\A(a)$ and $\beta \in \sigma_\A(b)$. 
		For any state $\omega^{vn}: \A^{vn} \rightarrow \Comp$, the universal representation $\Pi:\A\rightarrow \A^{vn}$ allows us to define the corresponding state $\omega$ on $\A$, setting $\omega := \omega^{vn} \circ \Pi$. Then if $\hat{C} \in \A^{vn}$ and $f:\A^{vn} \rightarrow \A$ is a function we can write that:
		\begin{equation*}
			\begin{split}
				\omega(f(\hat{C})) &= \omega^{vn} ( \Pi( f(\hat{C}) )  ) \\
				&= \int_{\sigma_{\A^{vn}}(\hat{C})} [\Pi \circ f](\mathbf{x}) \mu_{\omega^{vn}}^{(\hat{C})}(d\mathbf{x}).
			\end{split}
		\end{equation*}
		Suppose that $\hat{C}$ is the operator of \textsc{Step 1}. By theorem \ref{JSM} we can also write that:
		\begin{equation*}
			\begin{split}
				\omega(f(\hat{C})) =  \int_{\sigma_{\A^{vn}}(\hat{A}_\omega) \times \sigma_{\A^{vn}}(\hat{B}_\omega)} [\Pi \circ f](\mathbf{x}) \mu_{\omega^{vn}}^{(\hat{C})}(d\mathbf{x}) 
			\end{split}
		\end{equation*}
		Form what we found in \textsc{Step 2}, we know that $\sigma_{\A^{vn}}(\hat{A}_\omega) \times \sigma_{\A^{vn}}(\hat{B}_\omega) \subset \sigma_\A(a) \times \sigma_\A(b)$. Let us now choose the state $\omega'^{vn}$ such that $\mu_{\omega'^{vn}}^{(\hat{C})}(d\mathbf{x}) = \delta_\gamma (d\mathbf{x})$ for some $\gamma = (\alpha_1,\beta_1)$, which is always possible for a single random variable. Now if we set $f = f_1$ as in \textsc{Step 1}, we can write that:
		\begin{equation*}
			\begin{split}
				\omega'(a) &= \int_{\sigma_{\A^{vn}}(\hat{A}_\omega) \times \sigma_{\A^{vn}}(\hat{B}_\omega)} \Pi(f_1(\mathbf{x})) \mu_{\omega'^{vn}}^{(\hat{C})}(d\mathbf{x}) \\
				&= \int_{\sigma_{\A^{vn}}(\hat{A}_\omega) \times \sigma_{\A^{vn}}(\hat{B}_\omega)} \Pi(f_1(\mathbf{x})) \delta_\gamma (d\mathbf{x}) \\
				&= \int_{\sigma_{\A^{vn}}(\hat{A}_\omega) } \Pi(x) \delta_{\alpha_1} (dx) \\
				&= \int_{\sigma_{\A}(a)} \alpha \mu_{\omega'}^{(a)}(d\alpha)
			\end{split}
		\end{equation*}
		and so $\mu_{\omega'}^{(a)}(d\alpha) = \delta_{\alpha_1}(d\alpha)$ with $\alpha_1 = f_1(\gamma) \in \sigma_{\A^{vn}}(\hat{A}_\omega) \subset \sigma_{\A}(a)$. The same holds for $b$ setting $f = f_2$. Thus for the same state, we have two delta-like probability measures for commuting observables: this is simply a proof of the existence of dispersion free states for abelian $C^*$-algebra.
		
		\textsc{Step 4: $\inf_\omega [H_\omega(a,\varepsilon) + H_\omega(b,\delta)] = 0 \Leftrightarrow [a,b]=0$.}
		
		Let $\omega'$ be the state found in \textsc{Step 3} and $\varepsilon = \{E_i\}_{i \in I}$, $\delta = \{F_j\}_{j \in J}$ two partitions. Using $\omega'$, clearly $C(\varepsilon,\delta):= \inf_\omega [H_\omega(a,\varepsilon) + H_\omega(b,\delta)] = 0$ for any $(\varepsilon,\delta)$-partition. Suppose that $C(\varepsilon,\delta) = 0$ even if $[a,b] \neq 0$ for any $(\varepsilon,\delta)$-partition. If this is possible, since the Shannon entropy is always non-negative, the only possibility to have $C(\varepsilon,\delta)=0$ is to have a delta-like spectral measure for both $a$ and $b$. In the state where this happens we have a common probability space for $a$ and $b$ (i.e. there exists a joint PVM). But this contradicts theorem \ref{JSM}. Note that we have no contradiction, if $C(\varepsilon,\delta) = 0$ just for some $(\varepsilon,\delta)$-partition. Indeed, we can always have $C(\varepsilon,\delta) = 0$ for any partition where the supports of the induced probability measures is completely contained in exactly one set of the partition, if this happens for both $a$ and $b$, namely, if $\mbox{supp} \mu_\omega^{(a)} \subset E_i \in \varepsilon$ and $\mbox{supp} \mu_\omega^{(b)} \subset F_j \in \delta$, then $C(\varepsilon,\delta) = 0$. Thus we can say that if for some partition $(\varepsilon,\delta)$
		\begin{equation*}
			H_\omega(a;\varepsilon) + H_\omega(b;\delta) \geqslant C(\varepsilon,\delta) > 0 
		\end{equation*} 
		then $[a,b] \neq 0$, and this concludes the proof.
	\end{proof}
	
	\subsection*{C - Proof of the corollary \ref{mycor}}
	
	\begin{proof}
		Let $\omega$ be a mixed state, hence it can be written as $\lambda\omega_1 + (1-\lambda)\omega_2$, for $\lambda \in (0,1)$. If $a \in \A$ is a self-adjoint element of the algebra and $\varepsilon$ is a given partition, from the definition of $p_i^{\omega,\varepsilon}$, we can see that:
		\begin{equation*}
			p_i^{\omega,\varepsilon} = \lambda p_i^{\omega_1,\varepsilon} + (1-\lambda) p_i^{\omega_2,\varepsilon}.
		\end{equation*}
		Because the entropy is a concave function, we can write that $H_\omega (a;\varepsilon) \geqslant \lambda H_{\omega_1}(a;\varepsilon) + (1-\lambda)H_{\omega_2}(a;\varepsilon)$. Hence if $a$ and  $b$ fulfil an entropic uncertainty relation for pure states, then it holds also for mixed states, with the same constant $C(\varepsilon,\delta)$. Indeed:
		\begin{equation*}
			\begin{split}
				H_\omega (a;\varepsilon) + H_\omega (b;\delta)&\geqslant \lambda [H_{\omega_1}(a;\varepsilon) + H_{\omega_1}(b;\delta)]  \\
				&\mspace{20mu}+ (1-\lambda)[H_{\omega_2}(a;\varepsilon) + H_{\omega_2}(b;\delta)] \\
				&\geqslant \lambda C(\varepsilon,\delta) + (1-\lambda)C(\varepsilon,\delta) = C(\varepsilon,\delta)
			\end{split}
		\end{equation*}
		which concludes the proof.
	\end{proof}

	\bibliographystyle{unsrt}
	\bibliography{bib-a2.bib}

\end{document}